\documentclass[twocolumn]{IEEEtran} 
\usepackage[T1]{fontenc}
\usepackage{color}
\usepackage{mathrsfs}
\usepackage{amsmath}
\usepackage{amsthm} 
\usepackage{amssymb}
\usepackage{graphicx}
\usepackage{cite}
\usepackage{makecell}
\usepackage{algorithm}
\usepackage{algorithmic}
\usepackage{fancyhdr}
\usepackage{diagbox}
\usepackage{bm}
\usepackage{float}
\usepackage{lettrine}
\usepackage[unicode=true,
bookmarks=true,bookmarksnumbered=true,bookmarksopen=true,bookmarksopenlevel=1,
breaklinks=false,pdfborder={0 0 0},pdfborderstyle={},backref=false,colorlinks=false]
{hyperref}
\hypersetup{
pdftitle={Your Title},
pdfauthor={Your Name},
pdfpagelayout=OneColumn, pdfnewwindow=true, pdfstartview=XYZ, plainpages=false}

\usepackage{multirow} 
\usepackage{xcolor}

\allowdisplaybreaks[4]

\usepackage[caption=false,font=footnotesize]{subfig}

\IEEEoverridecommandlockouts

\usepackage{geometry}
\renewcommand{\arraystretch}{1.25}
\usepackage{stfloats}
\newtheorem*{lemma}{Lemma}

\begin{document}

\title{\textcolor{black}{Robust and Secure Computation Offloading and
Trajectory Optimization for Multi-UAV MEC Against Aerial Eavesdropper}}
\author{\IEEEauthorblockN{Can Cui, Ziye Jia, \IEEEmembership{Member, IEEE}, Jiahao You, Chao Dong, \IEEEmembership{Member, IEEE},\\
Qihui Wu, \IEEEmembership{Fellow, IEEE}, and Zhu Han, \IEEEmembership{Fellow, IEEE}}
\thanks{This work was supported by the National Key R{\&}D Program of China under Grant 2022YFB3104502, in part by National Natural Science Foundation of China under Grant 62301251, in part by the open research fund of National Mobile Communications Research Laboratory, Southeast University (No. 2024D04), in part by the Aeronautical Science Foundation of China 2023Z071052007, in part by the Young Elite Scientists Sponsorship Program by CAST 2023QNRC001, and partially supported by NSF ECCS-2302469, CMMI-2222810, Toyota. Amazon and Japan Science and Technology Agency (JST) Adopting Sustainable Partnerships for Innovative Research Ecosystem (ASPIRE) JPMJAP2326. (\textit{Corresponding author: Ziye Jia.})}
\thanks{Copyright (c) 2025 IEEE. Personal use of this material is permitted. However, permission to use this material for any other purposes must be obtained from the IEEE by sending a request to pubs-permissions@ieee.org.}
\thanks{Can Cui, Jiahao You, Chao Dong and Qihui Wu are with the College of Electronic and Information Engineering, Nanjing University of Aeronautics and Astronautics, Nanjing 211106, China (e-mail: cuican020619@nuaa.edu.cn; yjiahao@nuaa.edu.cn; dch@nuaa.edu.cn; wuqihui@nuaa.edu.cn).}
\thanks{Ziye Jia is with the College of Electronic and Information Engineering, Nanjing University of Aeronautics and Astronautics, Nanjing 211106, China, and also with the National Mobile Communications Research Laboratory, Southeast University, Nanjing, Jiangsu, 211111, China (e-mail: jiaziye@nuaa.edu.cn).}
\thanks{Zhu Han is with the University of Houston, Houston, TX 77004 USA, and also with the Department of Computer Science and Engineering, Kyung Hee University, Seoul 446-701, South Korea (e-mail: hanzhu22@gmail.com).}
}
\maketitle

\pagestyle{fancy}
\fancyhf{}
\fancyhead[R]{\fontsize{9}{11}\selectfont \thepage}

\begin{abstract}
	The unmanned aerial vehicle (UAV) based multi-access edge computing (MEC) appears as a popular paradigm to reduce task processing latency. However, the secure offloading is an important issue when occurring aerial eavesdropping. Besides, the potential uncertainties in practical applications and flexible trajectory optimizations of UAVs pose formidable challenges for realizing robust offloading. In this paper, we consider the aerial secure MEC network including ground users, service unmanned aerial vehicles (S-UAVs) integrated with edge servers, and malicious UAVs overhearing transmission links. To deal with the task computation complexities, which are characterized as uncertainties, a robust problem is formulated with chance constraints. The energy cost is minimized by optimizing the connections, trajectories of S-UAVs and offloading ratios. Then, the proposed non-linear problem is tackled via the distributionally robust optimization and conditional value-at-risk mechanism, which is further transformed into the second order cone programming forms. Moreover, we decouple the reformulated problem and design the successive convex approximation for S-UAV trajectories. The global algorithm is designed to solve the sub-problems in a block coordinate decent manner. Finally, extensive simulations and numerical analyses are conducted to verify the robustness of the proposed algorithms, with just 2\% more energy cost compared with the ideal circumstance.
\end{abstract}

\begin{IEEEkeywords}
	Unmanned aerial vehicles, multi-access edge computing, data offloading, secure communication, robust optimization.
\end{IEEEkeywords}

\section{Introduction}\label{sec1}

\lettrine[lines=2]{W}{ITH} the development of the sixth generation of communication system, there emerge explosive computationally-intensive and delay-sensitive applications, which bring a critical issue in guaranteeing stable communication services and achieving efficient network performance for their high requirements for quality of service (QoS). In response to this issue, the multi-access edge computing (MEC) provides new opportunities to reduce the responding latency and satisfy the QoS demands for heterogeneous tasks \cite{GDSG-Liang,Joint-Hoang,Computation-Yang}. However, due to the unaffordable cost and insufficient coverage, it is unrealistic to achieve comprehensive deployment of MEC infrastructures simply by static ground base stations (BSs). Fortunately, the unmanned aerial vehicles (UAVs) based MEC can provide promising services for alleviating congestions and enhancing performances \cite{Toward-Bai,Energy-Pervez,Cooperative-Jia,Joint-Dong}. As an innovative candidate, the UAVs produce a marked effect in data collection, transmission and edge computing \cite{Adaptive-Wu,Computing-Wang}. By deploying edge servers to the user sides, the multiple UAVs assisted MEC technology can significantly compensate the shortcomings of ground infrastructures with a lower cost and better coverage.

Nevertheless, regardless of the advantages and superiorities of such a multi-UAV enabled MEC paradigm, there still exist many challenges related with reliability and security to be focused on, detailed as follows.
\begin{itemize}
	\item{\em Challenge 1:}  Due to the openness of the airspace, there exist potential eavesdroppers that may significantly threaten the security of the aerial MEC paradigm \cite{Learning-Based-Liu, Joint-Li}. Hence, it is significant to highlight the secure communication, and how to  inhibit the overhearing behaviors from malicious eavesdroppers becomes a tricky issue to ensure the secrecy performance \cite{Joint-Xu,Energy-Hua}.
	\item{\em Challenge 2:}  The dynamic and heterogeneous demands of tasks are time-varying, and it faces great uncertainties in practical applications since the task complexity is not always known. However, the traditional methods fail to take the uncertainties into account, resulting in unexpected delays and failures in QoS demands. Consequently, a robust algorithm is necessary for the worst scenarios to guarantee the performance of MEC services \cite{Robust-Fan}.
	\item{\em Challenge 3:} The energy cost is a key metric in the UAV-assisted MEC network due to the constrained battery capacities and limited service time \cite{Aiding-Yang}. However, it faces great time complexity to elaborately design the offloading strategies and the trajectories for the sake of energy saving with high dimensions, especially via learning based methods \cite{Frontiers-Yang}.
\end{itemize}

To deal with the above challenges, we creatively investigate a multi-UAV based secure MEC scenario which can be applied in smart cities and provide temporary communication and computation services in remote areas. In particular, multiple service unmanned aerial vehicles (S-UAVs) equipped with computing resources navigate in the area to serve as aerial BSs. A ground jammer (GJ) broadcasts jamming signals continuously to inhibit the eavesdropping unmanned aerial vehicle (E-UAV) which overhears the communication links. The uncertain computation complexities of tasks are taken into account, and the problem is formulated with regard to chance constraints, which are accordingly transformed into distributionally robust chance constraints (DRCCs) under the worst case. Moreover, we utilize the conditional value-at-risk (CVaR) mechanism to deal with the DRCCs, and the mixed integer non-linear programming (MINLP) problem is further decoupled into three sub-problems. To take advantages of the mobility of UAVs, the successive convex approximation (SCA) based method is exploited for trajectory optimization, and the problem is figured out via alternatively solving the three sub-problems by convex optimization toolkits. Furthermore, the numerical results show great superiorities and robustness of the designed algorithms.

In summary, the main contributions are concluded as follows.

\begin{itemize}
	\item An innovative framework considering a potential malicious eavesdropper is proposed. To prevent the overhearing from E-UAV and ensure the security of communication links, a GJ is deployed to broadcast interference signals. We aim to minimize the total energy cost via jointly optimizing the trajectories of S-UAVs, GU-UAV connection relationships, and partial offloading ratios. 
	\item We consider the potential uncertainties in task computation complexities, and the chance constraints are constructed under the uncertainty set based on their first and second order moment estimation information. To tackle the chance constraints, we firstly leverage the distributionally robust optimization (DRO) method and convert them into DRCCs. Furthermore, via the theory of CVaR, the DRCCs are transformed into mixed integer second order cone programming (MISOCP) forms.
	\item The reformulated MINLP problem is decomposed into three sub-problems by the primal decomposition. The sub-problems related with the offloading strategies and GU-UAV connection relationships are solved using optimization tools of CVX and MOSEK, respectively. Since the sub-problem concerning S-UAV trajectories is intractable, we further design the SCA method.
	\item To evaluate the proposed algorithms, numerical simulations are conducted under diverse circumstances. The robustness of the proposed method is verified against the potential uncertainties, and the superiority in optimization performance of the proposed algorithms is revealed compared with other baseline scheme.
\end{itemize}

The rest of this paper is arranged as follows. Related works are presented in Section \ref{sec2}. Section \ref{sec3} proposes the system model and problem formulation. The mathematical reformulation and decompositions, as well as algorithms are completed in Section \ref{sec4}. Simulations results and corresponding analyses are provided in Section \ref{sec5}. Finally, we draw conclusions in Section \ref{sec6}.

\section{Related Works}\label{sec2}
Boosting for their flexibility and adaptability, UAVs as aerial BSs have attracted tremendous attentions from academics. There has been a surge in works that deploy UAVs to provide MEC services for GUs. \cite{Stackelberg-Wang} studied an offloading problem in a multi-UAV based MEC system and designed a multi-round iterative game algorithm to enhance the user satisfaction. \cite{Multiagent-Lee} aimed to minimize the time averaged total system energy consumption in the UAV-based MEC network while maintaining stability using reinforcement learning. \cite{Service-Zheng} explored an iterative algorithm based on the block coordinate descent and SCA methods to minimize the task completion latency for GUs, in which multi-UAVs equipped with MEC servers acted as aerial BSs. The study in \cite{Cooperative-Liu} presented an MEC framework, in which UAVs cooperatively provided offloading services for GUs, and a deep reinforcement learning based algorithm was proposed for the optimal computation offloading and resource management policies. \cite{Joint-Liu} investigated a task computation data volume maximizing problem in the UAV enabled MEC system via the block coordinate descent method to jointly optimize the UAV trajectory, communication and computation resources allocation. \cite{Fairness-He} formulated a fairness based offloading problem in the multi-UAV enabled MEC network for the offloading strategy, selectivity, and UAV trajectories design. However, the unpredictable uncertainties are ignored in the above researches, which are stemmed from the limitations of practical applications and can lead to distinguished outage. A smaller subset of studies attempts to handle uncertainties. \cite{Adaptive-Chen} focused on the adaptive bitrate video streaming with uncertain content popularity in UAV assisted MEC networks and developed a DRO algorithm for risk-averse solutions. Moreover, the double uncertainties for the channel state information and task complexity in the multi-UAV assisted MEC scenario were taken into consideration in \cite{Robust-Li}, and a multi-agent proximal policy optimization was designed to jointly optimize the UAV trajectory, task partition, as well as resource allocation. While the uncertainties are addressed in these works, they fail to integrate robustness with secure communication for the multi-UAV enabled MEC framework, which is exposed in the open airspace and vulnerable to the malicious eavesdropping.

Owing to the existence of potential eavesdroppers which poses a great menace, the security of communication between GUs and UAVs is a critical factor in such an MEC system. To provide secure offloading services for GUs, a couple of works have been investigated. For instance, \cite{Resource-Lu} jointly optimized the transmission power, time slot allocation factor, and UAV trajectory by integrating block coordinate descent and SCA approaches for the proposed secure data transmission scheme in the UAV assisted maritime MEC system with a flying eavesdropper. The authors in \cite{DDQN-Ding} studied the secure transmission in UAV based MEC systems to maximize the average secure computing capacity and a double-deep Q-learning method was proposed. \cite{Task-Zhang} proposed a joint optimization problem for offloading decision, resource allocation, and trajectory planning in the multi-UAV assisted MEC system to maximize the secure calculation capacity of GUs. \cite{Collaborative-Ding} considered a UAV-enabled MEC secure system against aerial eavesdropping to minimize the energy consumption. \cite{Secure-Zhao} formulated a joint optimization problem for secure offloading with an aerial eavesdropper in a multi-UAV based MEC network, which was high-dimensional nonlinear mixed-integer NP-hard and solved via the deep reinforcement learning based approach. \cite{Secure-Li} designed a multi-agent reinforcement learning based method for secure offloading and resource allocation in the multiple UAVs assisted MEC system with eavesdroppers. The authors in \cite{Online-Ding} considered the information security against eavesdroppers, and formulated the problem to maximize the secure computation efficiency by optimizing the offloading decision and resource management based on the deep reinforcement learning and SCA methods. While these available literatures focus on mitigating eavesdropping threats in the UAV assisted MEC, the task models are simplified with deterministic complexities and the robustness in the face of unpredictable task demands cannot be guaranteed.

In view of prior works, the available researches have been studied towards the efficient offloading and secure communication in the multiple UAVs enabled MEC networks. However, in these studies, the unpredictable computation complexities are mostly ignored, which may pose great challenges for efficient offloading. It leaves a critical gap in the scenarios where both the security and adaptability to dynamic tasks are required. Therefore, in the secure communication scenario with the potential aerial eavesdropper, we capture the computation uncertainties for a robust solution to against the above limitations. By elaborating on the optimization of offloading strategies and UAV trajectories, the flexibility of UAVs are taken full considerations and the high QoS can be guaranteed.

\begin{table}[tbp]
	\renewcommand\arraystretch{1.55}
	\caption{KEY NOTATIONS}\label{notation}
	\begin{center}
	\begin{tabular}{|c|l|}
	\hline
	Parameter & Description \\
	\hline
	$L_i(t)$  & Data length of task $i$ in time slot $t$.\\
	\hline
	$c_i(t)$ &  Computation complexity of task $i$ in time slot $t$.\\
	\hline
	$\Delta_i(t)$ & Random task computation complexity estimation error.\\
	\hline
	$\mathcal{P}$ & Uncertainty set for $\Delta_i(t)$.\\
	\hline
	$E_m^{fly}(t)$ & Propulsion energy cost for S-UAV $m$ in time slot $t$.\\
	\hline
	$E_{i}^{l,c}(t)$ & \makecell[l]{Computation energy cost for GU $i$ to process task locally \\in time slot $t$.}\\
	\hline
	$E_{i,m}^{e,d}(t)$ & \makecell[l]{Transmission energy cost from GU $i$ to S-UAV $m$ \\in time slot $t$.}\\
	\hline
	$E_{i,m}^{e,c}(t)$ & \makecell[l]{Computation energy cost for S-UAV $m$ to process task $i$ \\in time slot $t$.}\\
    \hline
	$T_{i}^{l,c}(t)$ & \makecell[l]{Computation latency for GU $i$ to process task locally \\in time slot $t$.}\\
	\hline
    $T_{i,m}^{e,d}(t)$ & \makecell[l]{Transmission latency from GU $i$ to S-UAV $m$ \\in time slot $t$.}\\
	\hline
	$T_{i,m}^{e,c}(t)$ & \makecell[l]{Computation latency for S-UAV $m$ to process task $i$ \\in time slot $t$.}\\
	\hline
	\hline
	Variable & Description \\
	\hline
	$w_m(t)$ & Location of S-UAV $m$ in time slot $t$.\\
	\hline
	$\lambda_{i,m}(t)$ & Connection between GU $i$ and S-UAV $m$ in time slot $t$.\\
	\hline
	$\rho_i(t)$ & Offloading ratio from GU $i$ in time slot $t$.\\
	\hline
	\end{tabular}
	\end{center}
\end{table}

\section{System Model}\label{sec3}

In this section, we propose a multi-UAV based secure MEC scenario in Section \ref{subsection-1}. The task model characterized with random computation complexities is provided in Section \ref{subsection-2}. To take full advantages of the flexibility of UAVs, we analyze their mobility models in Section \ref{subsection-3} to navigate in the area. The secure communication model and computation model are illustrated in Sections \ref{subsection-4} and \ref{subsection-5}, respectively. The problem formulation to minimize the energy cost concerning chance constraints is put forward in Section \ref{subsection-6}. Besides, for clarity, the main notations are shown in Table \ref{notation}.

\subsection{Multi-UAVs based MEC Scenario}\label{subsection-1}

We focus on a multiple S-UAVs assisted MEC network with an E-UAV, as shown in Fig. \ref{Fig-1} to offer offloading services in infrastructure-limited remote areas or rescue missions. Specifically, $M$ S-UAVs are equipped with computing resources to provide offloading services for $I$ GUs with heterogeneous tasks. Let $\mathcal{I}=\{1,2,\cdots,i,\cdots,I\}$, $i\in\mathcal{I}$ denote the set of GUs, and $\mathcal{M}=\{1,2,\cdots,m,\cdots,M\}$, $m\in\mathcal{M}$ denote the set of S-UAVs. Due to the openness of airspace, an E-UAV can overhear the communication link and intercept the offloading signals out of malice in practical applications. To guarantee the security of the network, a GJ interferes with the E-UAV by broadcasting jamming signals. Since the S-UAVs have the prior knowledge of the interference signals from the GJ, they are immune to the jamming signals. However, the existence of GJ is not exposed to the E-UAV in advance. Thus, the received signals from both GUs and the GJ for the E-UAV are all treated as valid signals \cite{Online-Ding}. In this way, the overhearing behavior from the E-UAV can be inhibited by lowering its eavesdropping rate, thereby enhancing the secrecy transmission rate. To capture the dynamically time-varying characteristics for both UAVs and GUs, the total process is divided into $T$ time slots, where the index set for time slots is defined as $t\in\mathcal{T}=\{1,2,\cdots,T\}$. The time duration of the time slot is denoted by $\tau$.

\begin{figure}
	\centering{\includegraphics[width=0.98\linewidth]{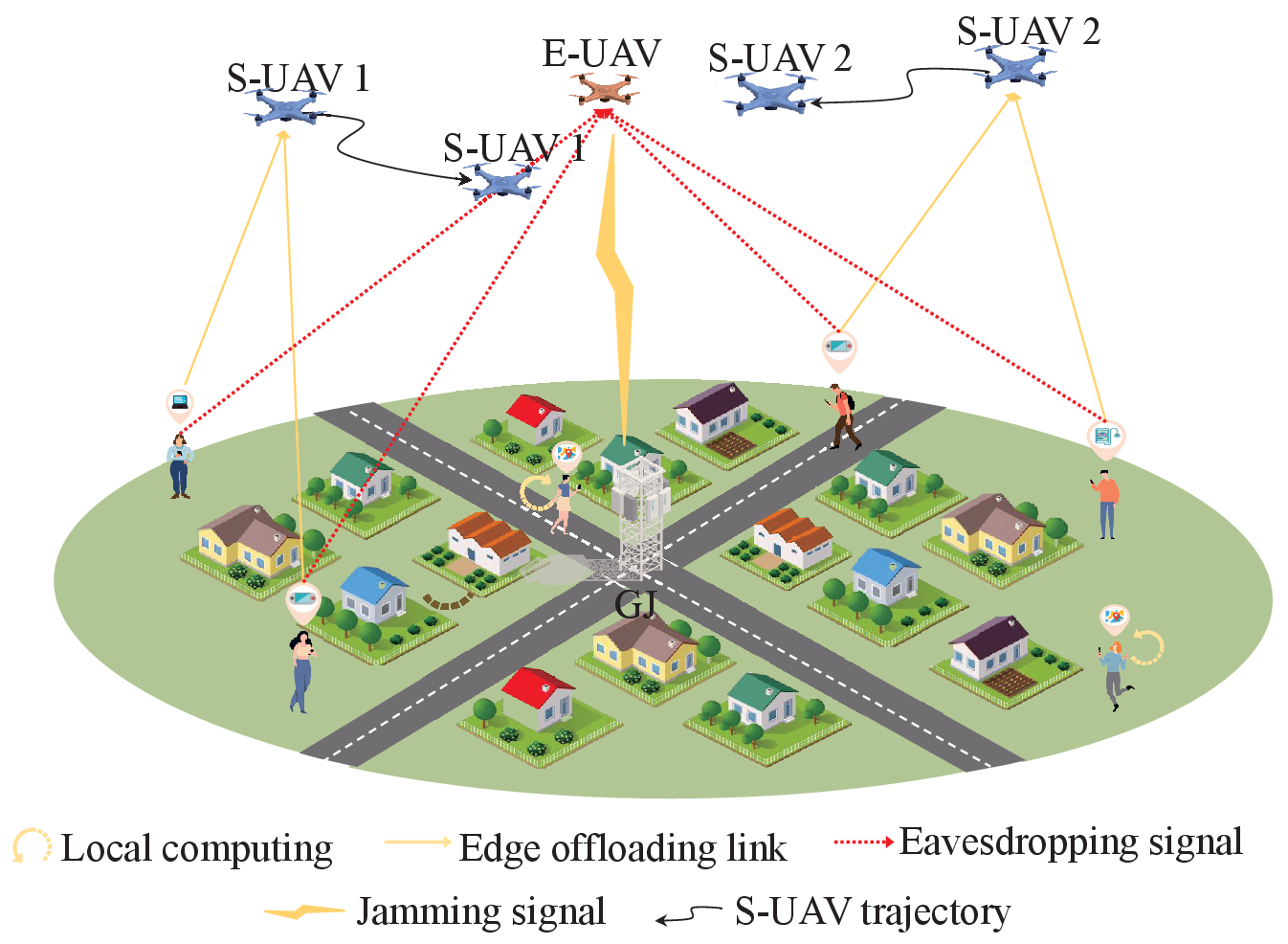}}
	\caption{Multi-UAV based secure MEC scenario.}
	\label{Fig-1}
\end{figure}

\subsection{Task Model with Random Complexity}\label{subsection-2}
Each GU has a computation task to be processed in time slot $t$, denoted as $(L_i(t),c_i(t))$, where $L_i(t)$ represents the data length and $c_i(t)$ represents the computation complexity in unit of CPU cycles for processing 1 bit data. However, due to the limitation of practical systems, the complexities of tasks are uncertain which vary with the types of input data \cite{Device-Masoudi, Asynchronous-Liang}. We evaluate the computation complexities, by $c_i(t)$ as
\begin{equation}
	c_i(t)=\bar{c_i}(t)+\Delta_i(t),\forall i\in \mathcal{I}, t\in \mathcal{T}.
\end{equation}
Wherein, $\bar{c_i}(t)$ represents the estimated computation complexity of task $i$ at time slot $t$, and $\Delta_i(t)$ indicates the corresponding estimation error, which is assumed to follow an unknown distribution $\mathbb{P}$. It is formidable to measure the accurate value or possible distribution of such potential uncertainties \cite{Robust-Fan}. Instead, the first and second-order moment value can be easily obtained via historical statical information, denoted as $\mu_i(t)$ and $\sigma_i^2(t)$, respectively. Accordingly, we have
\begin{equation}
	\mathbb{E_P}(\Delta_i(t))=\mu_i(t),\forall i\in \mathcal{I}, t\in \mathcal{T},
\end{equation}
and
\begin{equation}
	\mathbb{D_P}(\Delta_i(t))=\sigma_i^2(t),\forall i\in \mathcal{I}, t\in \mathcal{T}.
\end{equation}
Wherein, $\mathbb{E_P}(\cdot)$ represents the mean value and $\mathbb{D_P}(\cdot)$ denotes the variance. Therefore, the uncertainty set $\mathcal{P}$ of the random parameter $\Delta_i(t)$ is constructed to describe all of their possible distributions of $\mathbb{P}$, where $\mathbb{P}\in\mathcal{P}$.

Since each GU selects at most one S-UAV for offloading assistance during each time slot, let binary variable $\lambda_{i,m}(t)\in\{0,1\}$ represent the index for connection relationship between GU $i$ and S-UAV $m$ at time slot $t$ \cite{Computational-Khalid}. If GU $i$ is connected to S-UAV $m$ at time slot $t$, $\lambda_{i,m}(t)=1$, and otherwise, $\lambda_{i,m}(t)=0$. Considering the accessing constraint, each GU can offload its data to at most one S-UAV during each time slot, i.e.,
\begin{equation}\label{c1}
	\sum_{m=1}^M\lambda_{i,m}(t)\leq 1, \forall i\in \mathcal{I},t\in\mathcal{T}.
\end{equation}

Moreover, due to the limitation of antennas as well as MEC servers, each S-UAV is assumed to connect to at most $M_{max}$ GUs for task processing simultaneously during each time slot, i.e.,
\begin{equation}\label{c2}
	\sum_{i=1}^I\lambda_{i,m}(t)\leq M_{max}, \forall m\in \mathcal{M},t\in\mathcal{T}.
\end{equation}
\subsection{Mobility Model of UAV}\label{subsection-3}
The Cartesian coordinate system is applied to describe the locations of both GUs and UAVs. The position of GU $i$ is represented as $w_i=(x_i,y_i),\forall i\in \mathcal{I}$ and the location of the GJ is $w_{j}=(x_{j},y_{j})$. For the sake of energy conservation, we assume each UAV including both the E-UAV and S-UAVs remains flying at the same height to avoid frequent ascending and descending. The horizontal coordinate of S-UAV $m$ at time slot $t$ is given by $w_m(t)=(x_m(t),y_m(t)),\forall m\in \mathcal{M}$ and the flying height of S-UAVs is $h_{s}$. Then, the Euclidean distance between S-UAV $m$ and GU $i$ at time slot $t$ is obtained by
\begin{equation}
	\begin{split}
		d_{i,m}(t)=\sqrt{(x_m(t)-x_i)^2+(y_m(t)-y_i)^2+h^2_{s}},\\
		\forall i\in \mathcal{I},m\in \mathcal{M},t\in \mathcal{T}.
	\end{split}
\end{equation}

Moreover, the flying trajectory of the E-UAV is known, whose position at time slot $t$ is represented as $w_{e}(t)=(x_{e}(t),y_{e}(t))$. $h_e$ denotes the height of the E-UAV. The distance between GU $i$ and the E-UAV at time slot $t$ is calculated as
\begin{equation}
	\begin{split}
		d_{i,e}(t)=\sqrt{(x_{e}(t)-x_i)^2+(y_{e}(t)-y_i)^2+h^2_e},\\
		\forall i\in \mathcal{I},t\in \mathcal{T},
	\end{split}
\end{equation}
and the distance between the E-UAV and GJ at time slot $t$ is
\begin{equation}
	d_{e,j}(t)=\sqrt{(x_{e}(t)-x_{j})^2+(y_{e}(t)-y_{j})^2+h^{2}_e},
	\forall t\in \mathcal{T}.
\end{equation}

We denote $X_{min}$ and $X_{max}$ as the horizontal bounds of the specific area and $Y_{min}$ and $Y_{max}$ as the vertical bounds. The inherent constraints for the positions of S-UAVs are
\begin{equation}\label{c3}
	X_{min}\leq x_m(t) \leq X_{max},\forall m\in \mathcal{M},t\in \mathcal{T},
\end{equation}
and
\begin{equation}\label{c4}
	Y_{min}\leq y_m(t) \leq Y_{max},\forall m\in \mathcal{M},t\in \mathcal{T}.
\end{equation}
Moreover, the S-UAVs have fixed starting and ending points, denoted as
\begin{equation}\label{c5}
	w_{m}(1)= w_{m}^{ini},\forall m\in\mathcal{M},
\end{equation}
and
\begin{equation}\label{c6}
	w_{m}(T)= w_{m}^{fin},\forall m\in\mathcal{M},
\end{equation}
where $w_{m}^{ini}$ and $w_{m}^{fin}$ are the initial and final locations for S-UAV $m$, respectively.

Since the tasks generated from GUs are heterogeneous and time-varying, it is necessary to take advantages of the mobility of S-UAVs and provide better services for GUs. During each slot, S-UAVs are assumed to navigate in the area with a fixed velocity $v_0$. Referring to lecture \cite{Energy-Zeng}, we employ the rotary wing UAV propulsion power model. The acceleration time of S-UAVs is ignorable \cite{Dynamic-Xu}, and the traveling time of S-UAV $m$ during time slot $t$ is calculated as
\begin{equation}
	T_m^{fly}(t)=\frac{||w_m(t)-w_m(t-1)||}{v_0},\forall m\in\mathcal{M},t\in\mathcal{T},
\end{equation}
where
\begin{equation}\label{c7}
	||w_m(t)-w_m(t-1)||_\leq v_0 \tau, \forall m\in\mathcal{M}, t\in\mathcal{T}.
\end{equation}
Moreover, the power consumption for S-UAV $m$ in its flight process with the speed of $v_0$ consists of blade profile, induced and parasite \cite{Energy-Zeng}, i.e.,
\begin{equation}
	\begin{split}
		p_m^{fly}=&\underbrace{ P_1 \left(1+\frac{3v_0^2}{v_{bla}^2}\right)}_{\text{blade profile}}
				+\underbrace{ P_2 \left(\sqrt{1+\frac{v_0^4}{4v_{rot}^4}}-\frac{v_0^2}{2v_{rot}^2}\right)^\frac{1}{2} }_{\text{induced}}\\
				+&\underbrace{ \frac{1}{2} g \rho_0 s_0 a_0 v_0^3 }_{\text{parasite}},\qquad\forall m\in\mathcal{M},
	\end{split}
\end{equation}
where $P_1$ and $P_2$ denote the blade profile power and induced power under hovering status, respectively. $v_{bla}$ represents the velocity of the rotor blade. $v_{rot}$ is the mean rotor speed. $g$ is the drag ratio. $\rho_0$ and $s_0$ are deemed as the air density and the rotor solidity, respectively. $a_0$ denotes the rotor disc area. Specifically, the power consumption for hovering is obtained by substituting $v_0=0$ into $p_m^{fly}$, i.e.,
\begin{equation}
	\begin{split}
		p_m^{hov}=P_1+P_2,\forall m\in\mathcal{M}.
	\end{split}
\end{equation}
During each time slot, the propulsion energy consumption for S-UAV $m$ consists of the flying energy cost and hovering energy cost \cite{Dynamic-Xu}, calculated as 
\begin{equation}
	\begin{split}
			E_m^{fly}(t)=\underbrace{p_m^{fly} T_m^{fly}(t)}_{\text{flying}}+\underbrace{p_m^{hov} (\tau-T_m^{fly}(t))}_{\text{hovering}}&,\\
			\forall m\in\mathcal{M}, & t\in\mathcal{T}.
	\end{split}
\end{equation}

\subsection{Communication Model}\label{subsection-4}
Owing to the limited computing resources, it may be unaffordable for GUs to complete their tasks locally within the desired time period. In this case, we assume the computing tasks generated from GUs are divisible, and GUs are allowed to offload their tasks to S-UAVs for further process. In other words, the partial offloading model is followed, in which S-UAVs with edge servers are enabled to process the offloaded data in parallel, while the remaining data is processed locally at the GU sides \cite{Partial-Niu}. Let continuous variable $\rho_i(t)\in [0,1]$ represent the partial data generated from GU $i$ at time slot $t$ and offloaded to the edge servers, which satisfies
\begin{equation}\label{c8}
	\rho_i(t)\leq\sum_{m=1}^M \lambda_{i,m}(t),\forall i\in \mathcal{I}, t\in \mathcal{T}.
\end{equation}

We consider the communication model between GUs and S-UAVs as line of sight channel (LoS) model, which is modeled as free space path loss \cite{Service-Jia}. The channel gain between GU $i$ and S-UAV $m$ at time slot $t$ is calculated as
\begin{equation}
	g^{up}_{i,m}(t)=g_0 d^{-2}_{i,m}(t),\forall i\in \mathcal{I}, m\in \mathcal{M}, t\in \mathcal{T},
\end{equation}
where $g_0$ denotes the path loss at the reference distance $d_0=1$m. Moreover, for better interference management, the orthogonal frequency division multiple access (OFDMA) technology is employed for data offloading \cite{Enhancing-Chen}. The signal-to-interference-and-noise-ratio (SINR) is given by
\begin{equation}
	\gamma_{i,m}^{up}(t) = \frac{p_0g^{up}_{i,m}(t)}{n_0B_0}, \forall i\in \mathcal{I}, m\in \mathcal{M}, t\in \mathcal{T},
\end{equation}
where $B_0$ represents the bandwidth of the communication channel, and $p_0$ is the transmission power of GUs. $n_0$ is the power density of the channel additive noise.

As for the E-UAV, the channel gain between GU $i$ and E-UAV is
\begin{equation}
	g^{eav}_{i,e}(t)=g_0 d^{-2}_{i,e}(t),\forall i\in \mathcal{I}, t\in \mathcal{T},
\end{equation}
and the channel gain between the E-UAV and GJ is
\begin{equation}
	g^{jam}_{e,j}(t)=g_0 d^{-2}_{e,j}(t),\forall t\in \mathcal{T}.
\end{equation}
The deployed GJ continuously broadcasts the interference jamming signals to resist the overhearing, which become disruptive to the E-UAV and have an negative impact on its eavesdropping rate. Since it is difficult to distinguish whether the signals are from the GU or GJ \cite{Online-Ding}, the SINR for the E-UAV is
\begin{equation}
	\gamma_{i,e}^{eav}(t) = \frac{p_0g^{eav}_{i,e}(t)}{p^j g^{jam}_{e,j}(t)+n_0B_0},\forall i\in \mathcal{I}, t\in \mathcal{T}.
\end{equation}
Wherein, $p^j$ is the interference signal power from the GJ. In light of the Shannon formula, the uplink transmission rate from GU $i$ to S-UAV $m$ is given as
\begin{equation}
	r_{i,m}^{up}(t)=B_0\log_2(1+\gamma^{up}_{i,m}(t)),\forall i\in \mathcal{I}, m\in \mathcal{M}, t\in \mathcal{T}.
\end{equation}
Moreover, the achievable rate from GU $i$ to the E-UAV is
\begin{equation}
	r_{i,e}^{eav}=B_0\log_2(1+\gamma_{i,e}^{eav}(t)),\forall i\in \mathcal{I}, t\in \mathcal{T}.
\end{equation}
Consequently, the secure offloading rate \cite{Secure-Zhou} with the presence of the E-UAV is 
\begin{equation}
	r_{i,m}^{sec}(t)=\left[r_{i,m}^{up}(t)-r_{i,e}^{eav}\right]^+,\forall i\in \mathcal{I}, m\in \mathcal{M}, t\in \mathcal{T},
\end{equation}
where $[\cdot]^+\triangleq \max\{\cdot,0\}$. The transmission delay for GU $i$ to S-UAV $m$ at time slot $t$ is
\begin{equation}\label{e3}
	T^{e,d}_{i,m}(t)=\frac{\lambda_{i,m}(t)\rho_i(t)L_{i}(t)}{r_{i,m}^{sec}(t)},\forall i\in \mathcal{I}, m\in \mathcal{M}, t\in \mathcal{T},
\end{equation}
and the transmission energy cost is calculated as
\begin{equation}
	E^{e,d}_{i,m}(t)=p_0T^{e,d}_{i,m}(t),\forall i\in \mathcal{I}, m\in \mathcal{M}, t\in \mathcal{T}.
\end{equation}
\subsection{Computation Model}\label{subsection-5}
Integrated with the MEC technology, the tasks can be partially offloaded and the computation model consists of local computing performed by GUs themselves and edge offloading accomplished on S-UAVs. The detailed computation models are analyzed as follows.

\subsubsection{Local Computing}

We define the variable $f_g$ to represent the CPU frequency of GUs and $\varepsilon_g$ to represent the effective switched capacitance related to the architecture of MEC servers of GUs. For the local computing part of task $i$, the time latency task processing is derived as
\begin{equation}\label{e1}
	T^{l,c}_i(t)=\frac{(1-\rho_i(t))c_i(t)L_i(t)}{f_g},\forall i\in \mathcal{I}, t\in \mathcal{T},
\end{equation}
and the expected energy consumption for GU $i$ is
\begin{equation}
	\begin{split}
		&E^{l,c}_i(t)=\mathbb{E}\left\{\varepsilon_g(1-\rho_i(t))c_i(t)L_i(t)f^2_g\right\}\\
		&=\varepsilon_g(1-\rho_i(t))(\bar{c}_i(t)+\mu_i(t))L_i(t)f^2_g,\forall i\in \mathcal{I}, t\in \mathcal{T}.
	\end{split}
\end{equation}
Due to the uncertain computation complexity, we formulate a chance constraint for the latency, in which the computation process for each task must be completed within the time slot to avoid outdated data. In other words, the latency for processing task $i$ in time slot $t$ should not be larger than the duration $\tau$ with a probability of $\alpha$. For the local computing mode, we obtain
\begin{equation}\label{c9}
	\mathbf{Pr}_{\mathbb{P}}\left\{T^{l,c}_{i}(t)\leq \tau \right\}\geq\alpha,\forall i\in \mathcal{I},t \in \mathcal{T},
\end{equation}
where $\mathbf{Pr}_{\mathbb{P}}$ indicates the chance constraint is formulated under the unknown distribution $\mathbb{P}$ for uncertain computation complexity estimation error $\Delta_i(t)$.

\subsubsection{UAV-based Edge Offloading}As for tasks from GU $i$ which are partially offloaded to S-UAV $m$, the computation latency is given by
\begin{equation}\label{e2}
	T^{e,c}_{i,m}(t)=\frac{\lambda_{i,m}(t)\rho_i(t)c_i(t)L_i(t)}{f_u},\forall i\in \mathcal{I}, m\in \mathcal{M}, t\in \mathcal{T},
\end{equation}
where $f_u$ is the CPU frequency of UAVs \cite{Joint-You}. Moreover, let $\varepsilon_u$ denote the energy cost coefficient of S-UAVs. The expected computing energy consumption for S-UAV $m$ to process the task $i$ at time slot $t$ is calculated as
\begin{equation}
	\begin{split}
		E^{e,c}_{i,m}(t)=&\mathbb{E}\left\{\lambda_{i,m}(t)\varepsilon_u\rho_i(t)c_i(t)L_i(t)f^2_u\right\}\\
		=&\lambda_{i,m}(t)\varepsilon_u\rho_i(t)(\bar{c}_i(t)+\mu_i(t))L_i(t)f^2_u,\\
		&\qquad\qquad\qquad\forall i\in \mathcal{I}, m\in \mathcal{M}, t\in \mathcal{T}.
	\end{split}
\end{equation}
In this mode, the total latency for processing the offloading part to the S-UAV of a task includes both transmission and computation latency. Hence, when $\lambda_{i,m}(t)=1$, the chance constraint for latency in UAV-based offloading mode is
\begin{equation}\label{c10}
		\mathbf{Pr}_{\mathbb{P}}\left\{ T^{e,c}_{i,m}(t)+T^{e,d}_{i,m}(t)\leq \tau \right\}\geq\alpha,\forall  i\in \mathcal{I}, m\in \mathcal{M},t \in \mathcal{T}.
\end{equation}

The sum weighted energy consumption including S-UAVs and GUs at time slot $t$ consists of the transmission energy consumption for GUs, the computation energy consumption for GUs in the local computing model and for S-UAVs in the edge offloading model, and the propulsion energy consumption for the flight of S-UAVs, which can be denoted by
\begin{equation}
	\begin{split}
		E^{total}(t)=&\sum_{i=1}^I E^{l,c}_i(t)+\sum_{i=1}^I\sum_{m=1}^M E^{e,d}_{i,m}(t)\\
		+&\kappa\sum_{m=1}^M E^{fly}_m(t)+\kappa\sum_{i=1}^I\sum_{m=1}^M E^{e,c}_{i,m}(t),\forall t\in \mathcal{T},
	\end{split}
\end{equation}
where $\kappa$ is the weighted variable for a tradeoff among energy consumptions for S-UAVs and GUs.

\subsection{Problem Formulation}\label{subsection-6}

In this work, we emphasize the limited energy for both S-UAVs and GUs and our purpose is to minimize the total energy cost by optimizing the trajectories of S-UAVs $\bm{w_s}=\left\{w_m(t)|\forall m,\forall t\right\}$, connection strategy $\bm{\lambda}=\left\{\lambda_{i,m}(t)|\forall i,\forall m,\forall t\right\}$ and offloading ratio $\bm{\rho}=\left\{\rho_i(t)|\forall i,\forall t\right\}$. The problem $\textbf{P0}$ is mathematically formulated as
\begin{subequations}
	\begin{align}
		\textbf{P0:}\quad &\min_{\bm{w_s},\bm{\lambda},\bm{\rho}} \sum_{t=1}^T E^{total}(t)\nonumber\\
		\textrm{s.t.} \quad &(\ref{c1}),(\ref{c2}),(\ref{c3})-(\ref{c6}),(\ref{c7}),(\ref{c8}),(\ref{c9}),(\ref{c10}),\nonumber\\
		\quad &\rho_i(t) \in \left[ 0,1 \right],\forall i\in \mathcal{I},t \in \mathcal{T}\label{c11},\\
		\quad &\lambda_{i,m}(t)\in\{0,1\},\forall i\in \mathcal{I}, m\in \mathcal{M}, t \in \mathcal{T}\label{c12},
	\end{align}
\end{subequations}
which is obviously a non-convex problem subject to non-linear constraints. Besides, the inherent relationships among the coupled variables induce tremendous obstacles in solving MINLP problem $\textbf{P0}$.

\section{Problem Reformulation and Algorithm Design}\label{sec4}
The coupled and mixed-integer variables as well as the unavailable distribution of uncertain computation complexities for tasks lead to great challenges in solving $\textbf{P0}$. To address these issues, we first employ the DRO based mechanism in Section \ref{subsection-dro}. The DRCC problem is further transformed via CVaR in Section \ref{subsection-cvar} and the problem is decomposed in Section \ref{subsection-decomposition}. Then, we devote to optimizing the trajectories of S-UAVs, GU-UAV connection relationships, and offloading ratios. SCA is designed in Section \ref{subsection-sca} for the trajectory optimization. Besides, we design the global algorithm for solving the formulated non-convex problem in Section \ref{subsection-alg} integrating the CVaR mechanism and SCA.

\subsection{Distributionally Robust Optimization}\label{subsection-dro}
It lacks the distribution information for the random computation complexities $\Delta_i(t)$ in the chance constraints (\ref{c9}) and (\ref{c10}), which leads to troubles in figuring out solutions. Hence, the traditional methods relying on the known probability distributions cannot be directly applied. Fortunately, the DRO method can address this issue by optimizing the problem with chance constraints under the worst-case scenario within the uncertainty set $\mathcal{P}$ \cite{Distributionally-Jia}. Consequently, we turn to seeking a robust result and optimizing the problem under the worst distribution of the uncertain parameter $\Delta_i(t)$ via the DRO method. In detail, $\underset{\mathbb{P}\in\mathcal{P}}{inf}$ represents the lower bound of possibility for all potential distributions under the uncertainty set $\mathcal{P}$. Accordingly, the chance constraints in (\ref{c9}) and (\ref{c10}) are further transformed into the DRCCs by employing the DRO method as
\begin{equation}\label{drcc-1}
	\underset{\mathbb{P}\in\mathcal{P}}{inf}\quad\mathbf{Pr}_{\mathbb{P}}\left\{T^{l,c}_{i}(t)\leq \tau \right\}\geq\alpha,\forall i\in \mathcal{I},t \in \mathcal{T},
\end{equation}
and
\begin{equation}\label{drcc-2}
	\begin{split}
		\underset{\mathbb{P}\in\mathcal{P}}{inf}\quad\mathbf{Pr}_{\mathbb{P}}\left\{(T^{e,c}_{i,m}(t)+T^{e,d}_{i,m}(t))\leq \tau \right\}\geq\alpha,\\
		\forall i\in \mathcal{I}, m\in \mathcal{M},t \in \mathcal{T},
	\end{split}
\end{equation}
respectively, which can be optimized in a conservative and robust manner to struggle with the potential uncertainties. Consequently, the original problem $\textbf{P0}$ is turned into $\textbf{P1}$, i.e.,
\begin{equation}\nonumber
	\begin{split}
		\textbf{P1:}\quad &\min_{\bm{w_s},\bm{\lambda},\bm{\rho}} \sum_{t=1}^T E^{total}(t)\\
		\textrm{s.t.} \quad &(\ref{c1}),(\ref{c2}),(\ref{c3})-(\ref{c6}),\\
		&(\ref{c7}),(\ref{c8}),(\ref{c11}),(\ref{c12}),(\ref{drcc-1}),(\ref{drcc-2}).
	\end{split}
\end{equation}
By transforming the chance constraints into DRCCs, we are essentially seeking a robust solution against the most unfavorable distribution within the considered set. This approach provides a conservative yet reliable way to handle the potential uncertainties. While DRO ensures robustness under worst-case distributions, it is noticed that $\textbf{P1}$ is still complicated involving the coupled binary and continuous variables. Besides, with the uncertain parameter $\Delta_i(t)$, the DRCCs are still tricky to tackle. Thus, we exploit the CVaR mechanism for further convex optimization.

\subsection{CVaR-based Optimization}\label{subsection-cvar}
The lack of distribution information of uncertainties leads to the inability to tackle the reformulated $\textbf{P1}$ with DRCCs. To obtain an approximate and conservative estimation, we explore the CVaR based mechanism, which is a convex approximation to evaluate the chance constraints \cite{Distributionally-Ling, Distributionally-Cui}. Essentially, CVaR measures the expected loss beyond a certain confidence level and provides a more refined measure of risk compared to traditional methods \cite{Distributionally-Zymler,Value-Sarykalin,Optimization-Tyrrell}. In detail, for the loss function $\phi(\xi)$ concerning the random parameter $\xi$, the worst case constraint for CVaR constitutes a conservative estimation for the DRCC \cite{Distributionally-Ding}, given by
\begin{equation}
	\begin{split}
	\underset{\mathbb{P}\in \mathcal{P}}{sup}\quad \mathbb{P}-CVaR_{\alpha}\left( \phi \left( \xi \right) \right) \leq 0,\forall \mathbb{P}\in \mathcal{P}\\
	\Leftrightarrow \underset{\mathbb{P}\in \mathcal{P}}{inf}\quad \mathbb{P}\left\{\phi \left( \xi \right) \leq 0\right\}\geq \alpha ,
	\end{split}
\end{equation}
in which $\alpha$ is the expected safety factor. $\underset{\mathbb{P}\in \mathcal{P}}{sup}$ represents the upper bound of possible distribution under $\mathcal{P}$.
\begin{lemma}[]
	For the loss function $\phi(\xi)=\Theta\xi+\theta^0$, where $\Theta\in\mathbb{R}$ and $\theta^0\in\mathbb{R}$, the worst-case CVaR, i.e., $\underset{\mathbb{P}\in\mathcal{P}}{sup}\quad\mathbb{P}-CVaR_{\alpha}(\phi(\xi))$ can be further converted into the form of second order cone programming (SOCP), i.e.,
	\begin{equation}
		\begin{split}
			\underset{\beta,e,q,z,s}{inf}&\beta +\frac{1}{1-\alpha}\left( e+s \right) ,\\
			&e-\theta ^0+\beta +q-\Theta \mu -z>0,\\
			&e\geq 0, z> 0,\\
			&\begin{Vmatrix}
					q\\
					\Theta\sigma\\
					z-s
			\end{Vmatrix} \leq z+s,\\
		\end{split}
	\end{equation}
in which $\beta$, $e$, $q$, $z$, and $s$ are auxiliary variables. $\mu$ and $\sigma$ are the mean and standard deviation of random parameter $\xi$, respectively.   
\end{lemma}
\begin{proof}
	The detailed proof is in Appendix \ref{appendix}.
\end{proof}

For clarity, we use the variable sets to represent the groups of auxiliary variables for all time slots, i.e., $\varrho_{i,t,1}=\left\{\beta_{i,t,1},e_{i,t,1},q_{i,t,1},z_{i,t,1},s_{i,t,1}\right\}$ and $\varrho_{i,t,2}=\left\{\beta_{i,t,2},e_{i,t,2},q_{i,t,2},z_{i,t,2},s_{i,t,2}\right\}$, respectively, which help in reformulating the CVaR-based constraint into an SOCP-compatible form. Hence, leveraging $\textbf{Lemma}$, for all $i\in\mathcal{I}$ and $t\in\mathcal{T}$, the DRCCs in (\ref{drcc-1}) and (\ref{drcc-2}) can be turned into the MISOCP constraints in (\ref{drcc-3}) and (\ref{drcc-4}) as follows.

\begin{equation}\label{drcc-3}
	\begin{cases}
		&\underset{\varrho_{i,t,1}}{inf}\quad\beta_{i,t,1} +\frac{1}{1-\alpha}\left( e_{i,t,1}+s_{i,t,1}\right)\leq 0,\\
		&e_{i,t,1}-\theta_{i,1}(t)+\beta_{i,t,1}+q_{i,t,1}>\Theta_{i,1}(t) \mu_{i}(t) +z_{i,t,1},\\
		&e_{i,t,1}\geq 0, z_{i,t,1}> 0,\\
		&\begin{Vmatrix}
				q_{i,t,1}\\
				\Theta_{i,1}(t)\sigma_{i}(t)\\
				z_{i,t,1}-s_{i,t,1}
		\end{Vmatrix} \leq z_{i,t,1}+s_{i,t,1},\\
	\end{cases}
\end{equation}
and
\begin{equation}\label{drcc-4}
	\begin{cases}
		&\underset{\varrho_{i,t,2}}{inf}\quad\beta_{i,t,2} +\frac{1}{1-\alpha}\left( e_{i,t,2}+s_{i,t,2}\right)\leq 0,\\
		&e_{i,t,2}-\theta_{i,2}(t)+\beta_{i,t,2}+q_{i,t,2}>\Theta_{i,2}(t) \mu_{i}(t)+z_{i,t,2},\\
		&e_{i,t,2}\geq 0, z_{i,t,2}> 0,\\
		&\begin{Vmatrix}
				q_{i,t,2}\\
				\Theta_{i,2}(t)\sigma_{i}(t)\\
				z_{i,t,2}-s_{i,t,2}
		\end{Vmatrix} \leq z_{i,t,2}+s_{i,t,2},\\
	\end{cases}
\end{equation}
respectively, where
\begin{equation}
	\begin{cases}
		\Theta_{i,1}(t)=\frac{\left(1-\rho_i(t)\right)L_i(t)}{f_g},\forall i\in \mathcal{I}, t\in \mathcal{T},\\
		\Theta_{i,2}(t)=\frac{\sum\limits_{m=1}^M\lambda_{i,m}(t)\rho_i(t)L_i(t)}{f_u},\forall i\in \mathcal{I}, t\in \mathcal{T}.
	\end{cases}
\end{equation}

\begin{equation}
		\theta_{i,1}(t)=\frac{\left(1-\rho_i(t)\right)\bar{c}_i(t)L_i(t)}{f_g}-\tau,\forall i\in \mathcal{I}, t\in \mathcal{T},
\end{equation}
and
\begin{equation}
	\begin{split}
		\theta_{i,2}(t)=&\frac{\sum\limits_{m=1}^M\lambda_{i,m}(t)\rho_i(t)\bar{c}_i(t)L_i(t)}{f_u}\\
		&+\sum\limits_{m=1}^M \frac{\lambda_{i,m}(t)\rho_i(t)L_i(t)}{r^{sec}_{i,m}(t)}-\tau,\forall i\in \mathcal{I}, t\in \mathcal{T},
	\end{split}
\end{equation}
are obtained according to the complete expressions of $T^{l,c}_{i}(t)$ and $T^{e,c}_{i,m}(t)+T^{e,d}_{i,m}(t)$ in (\ref{e1}), (\ref{e2}), and (\ref{e3}), respectively. Therefore, by introducing the concept on the worst case of CVaR, (\ref{drcc-3}) and (\ref{drcc-4}) are further derived as MISOCP forms, and $\textbf{P1}$ is transformed into $\textbf{P2}$ without the uncertain parameters $\Delta_i(t)$, detailed as
\begin{equation}\nonumber
	\begin{split}
		\textbf{P2:}\quad &\min_{\bm{w_s},\bm{\lambda},\bm{\rho},\bm{\varrho_1},\bm{\varrho_2}} \sum_{t=1}^T E^{total}(t)\\
		\textrm{s.t.} \quad &(\ref{c1}),(\ref{c2}),(\ref{c3})-(\ref{c6}),(\ref{c7}),\\
		&(\ref{c8}),(\ref{c11}),(\ref{c12}),(\ref{drcc-3}),(\ref{drcc-4}),
	\end{split}
\end{equation}
where $\bm{\varrho_1}=\left\{\varrho_{i,t,1}|\forall i,\forall t\right\}$ and $\bm{\varrho_2}=\left\{\varrho_{i,t,2}|\forall i,\forall t\right\}$. By applying the DRO and CVaR mechanisms, we provide a practical way to handle the unexpected uncertainties in the worst case, and the problem is optimized in a risk-aware manner. However, the coupled binary variables $\bm{\lambda}$ and continuous variables $\bm{w_s},\bm{\rho},\bm{\varrho_1},\bm{\varrho_2}$ in \textbf{P2} are still hard to solve.

\subsection{Problem Decomposition}\label{subsection-decomposition}
With the aforementioned concept and theory on DRO and CVaR, the chance constraints in the original problem $\textbf{P0}$ are transformed in the form of MISOCP. However, it is observed that \textbf{P2} is still non-convex with the MISOCP constraints. To further deal with the coupled variables including binary and continuous variables, we first decompose $\textbf{P2}$ into three sub-problems, including: 1) $\textbf{P3}$ concerning optimizations for S-UAV trajectories $\bm{w_s}$; 2) $\textbf{P4}$ related to the GU-UAV connection relationships $\bm{\lambda}$; 3) $\textbf{P5}$ on offloading strategies $\bm{\rho}$ and auxiliary variables $\bm{\varrho_1}, \bm{\varrho_2}$. Specifically, with the given GU-UAV connection indicators $\bm{\lambda}$ and offloading strategies $\bm{\rho}$ as well as auxiliary variables $\bm{\varrho_1}$ and $\bm{\varrho_2}$, $\textbf{P3}$ is detailed as:
\begin{equation}\nonumber
	\begin{split}
		\textbf{P3:}\quad &\min_{\bm{w_s}} \sum_{t=1}^T E^{total}(t)\\
		\textrm{s.t.} \quad &(\ref{c3})-(\ref{c6}),(\ref{c7}),(\ref{drcc-4}),
	\end{split}
\end{equation}
which is non-convex. Besides, with the fixed variables $\bm{w_s}$, $\bm{\varrho_1}$, $\bm{\varrho_2}$, and $\bm{\rho}$, $\textbf{P4}$ is an integer programming problem, given as:
\begin{equation}\nonumber
	\begin{split}
		\textbf{P4:}\quad &\min_{\bm{\lambda}} \sum_{t=1}^T E^{total}(t)\\
		\textrm{s.t.} \quad &(\ref{c1}),(\ref{c2}),(\ref{c8}),(\ref{c12}),(\ref{drcc-4}).
	\end{split}
\end{equation}
Moreover, with the determined value of $\bm{w_s}$ and $\bm{\lambda}$, $\textbf{P5}$ concerning the offloading ratios is
\begin{equation}\nonumber
	\begin{split}
		\textbf{P5:}\quad &\min_{\bm{\varrho_1}, \bm{\varrho_2},\bm{\rho}} \sum_{t=1}^T E^{total}(t)\\
		\textrm{s.t.} \quad &(\ref{c8}),(\ref{c11}),(\ref{drcc-3}),(\ref{drcc-4}).
	\end{split}
\end{equation}
Note that the solutions of $\textbf{P5}$ can be computed using a standard convex optimization toolkit such as CVX, since $\textbf{P5}$ is a convex problem with SOCP constraints.
\subsection{SCA for S-UAV Trajectories}\label{subsection-sca}
$\textbf{P3}$ is a non-convex problem with regard to the distances between S-UAVs and GUs, i.e., $||w_m(t)-w_i||$ in constraint (\ref{drcc-4}), which brings up the non-convexity \cite{Joint-Xu}. To tackle the complicated problem, we develop an SCA based method for the sub-optimal trajectory design of S-UAVs. Specifically, the transmission latency from GUs to UAVs is detailed as
\begin{equation}
	\begin{split}
			T_{i,m}^{e,d}(t)=\frac{\lambda_{i,m}(t)\rho_i(t)L_i(t)}{\left[B_0\log_2(1+\frac{p_0g_0}{n_0B_0(||w_m(t)-w_i||^2+h_u^2)})-r_{i,e}^{eav}(t)\right]^+},\\
			\forall i\in \mathcal{I}, m\in \mathcal{M}, t\in \mathcal{T}.
	\end{split}
\end{equation}

\begin{figure*}[!ht]
	\begin{equation}\label{e51}
		T_{i,m}^{e,d}(t)\leq T_{i,m}^{up}(t) = \frac{\lambda_{i,m}(t)\rho_i(t)L_i(t)}{\left(B_0\log_2\left(\varUpsilon_{i,m}(t)\right)-r_{i,m}^{eav}(t)\right)^+}
		+\frac{\lambda_{i,m}(t)p_0g_0\rho_i(t)L_i(t)\ln2\left(\left\|w_m(t)-w_i\right\|^{2}-\left\|w_m^{(l)}(t)-w_i\right\|^{2}\right)}{n_0 \varUpsilon_{i,m}(t) \left(B_0\ln\left(\varUpsilon_{i,m}(t)\right)-r^{eav}_{i,e}(t)\ln2\right)^2\left(\left\|w_m^{(l)}(t)-w_i\right\|^{2}+h_u^{2}\right)^2}.
	\end{equation}
	\rule{\textwidth}{1pt}
\end{figure*}

Since $T_{i,m}^{e,d}(t)$ is a concave function concerning $||w_m(t)-w_i||^2$, for any given local point $w_m^{(l)}(t)$ in the $l$-th iteration and for all $i\in \mathcal{I}, m\in \mathcal{M}, t\in \mathcal{T}$, the global upper bound function of $T_{i,m}^{e,d}(t)$ can be replaced by $T_{i,m}^{up}(t)$ via utilizing the first-order Taylor expansion \cite{MEC-Zeng} as (\ref{e51}) , in which

\begin{equation}
	\begin{split}
		\varUpsilon_{i, m}(t) = 1+\frac{p_0g_0}{n_0B_0\left(\left\|w_m^{(l)}(t)-w_i\right\|^{2}+h_u^{2}\right)},\\
		\forall i\in \mathcal{I}, m\in \mathcal{M}, t\in \mathcal{T}.
	\end{split}
\end{equation}
\begin{algorithm}[!t]
	\caption{SCA for S-UAV trajectory optimization}\label{alg-sca}
	\begin{algorithmic}[1]
		\REQUIRE An initial feasible solution $\bm{w_s}^{(0)}$.
		\STATE \textit{Initialization:} Set the iteration index $l=0$.\label{line1-1}
		\REPEAT
		\STATE {Solve problem $\textbf{P6}$ with the given $\bm{w_s}^{(l)}$ to obtain the optimal solution $\bm{w_s}^{\star }$.}\label{line1-2}
		\STATE {Update $\bm{w_s}^{(l+1)} \leftarrow \bm{w_s}^{\star}$.}\label{line1-3}
		\STATE {Update the iteration index $l=l+1$.}\label{line1-4}
		\UNTIL {The result converges to the tolerant accuracy.}\label{line1-5}
		\ENSURE The trajectories of S-UAVs $\bm{w_s}$.
	\end{algorithmic}
\end{algorithm}

Thus, we can further introduce a more manageable function at a given local point in each iteration for estimation. Based on the above analyses, by taking the first-order Taylor expansion, (\ref{e51}) can be substituted into $\textbf{P3}$, and the problem can be rewritten as

\begin{subequations}
	\begin{align}
		\textbf{P6:}\quad &\min_{\bm{w_s}} \sum_{t=1}^T \widehat{E}(t)\nonumber\\
		\textrm{s.t.} \quad &(\ref{c3})-(\ref{c6}),(\ref{c7}),\nonumber\\
		&\sum\limits_{m=1}^M T_{i,m}^{up}(t) + \frac{\sum\limits_{m=1}^M\lambda_{i,m}(t)\rho_i(t)L_i(t)\bar{c}_i(t)}{f_u}<\varphi_i(t),\nonumber\\
		&\qquad\qquad\qquad\qquad\qquad\qquad\qquad\forall i\in \mathcal{I}, t\in \mathcal{T},
	\end{align}
\end{subequations}
in which 
\begin{equation}
	\begin{split}
		\varphi_i(t)=\tau + e_{i,t,2} + \beta_{i,t,2} + q_{i,t,2} - \Theta_{i,2}(t) \mu_{i}(t) -z_{i,t,2},\\
		\forall i\in \mathcal{I}, t\in \mathcal{T},
	\end{split}
\end{equation}
and
\begin{equation}
	\begin{split}
		\widehat{E}(t)=&\sum_{i=1}^I E^{l,c}_i(t)+\kappa\sum_{m=1}^M E^{fly}_m(t)+\kappa\sum_{i=1}^I\sum_{m=1}^M E^{e,c}_{i,m}(t)\\
		&+\sum_{i=1}^I\sum_{m=1}^M p_0T_{i,m}^{up}(t),\forall i\in \mathcal{I}, t\in \mathcal{T}.
	\end{split}
\end{equation}

Note that $\textbf{P6}$ is convex and can be efficiently solved by a standard convex toolkit. The detailed algorithm concerning the SCA mechanism is designed in Algorithm \ref{alg-sca}. The initial feasible solution $\bm{w_s}^{(0)}$ is given which satisfy the constraints of $\textbf{P6}$, and the iteration index $l=0$ is set to record the number of iterations (line \ref{line1-1}). Then, in each iteration, for the given $\bm{w_s}^{(l)}$, problem $\textbf{P6}$ is solved via CVX for the optimal solution $\bm{w_s}^{\star }$ (line \ref{line1-2}). $\bm{w_s}^{(l+1)}$ is updated with the obtained $\bm{w_s}^{\star}$ and the iteration index $l$ is updated for the next iteration (lines \ref{line1-3}-\ref{line1-4}). The iteration process continues until the result converges to the tolerance accuracy as provided in line \ref{line1-5}. We denote $l_{it}$ as the number of iterations. Since Algorithm \ref{alg-sca} is operated with the convex solver based on the interior-point method, the time complexity is mainly dependent on the number of S-UAVs $M$ and time slots $T$. Besides, the number of iterations $l_{it}$ has an impact as well. Accordingly, the time complexity of the proposed Algorithm \ref{alg-sca} is given by $\mathcal{O}\left((MT)^{3.5}l_{it}\log_2{\epsilon^{-1}}\right)$, where $\epsilon$ is the accepted duality gap \cite{Delay-Li}.

\begin{figure}
	\centering{\includegraphics[width=0.95\linewidth]{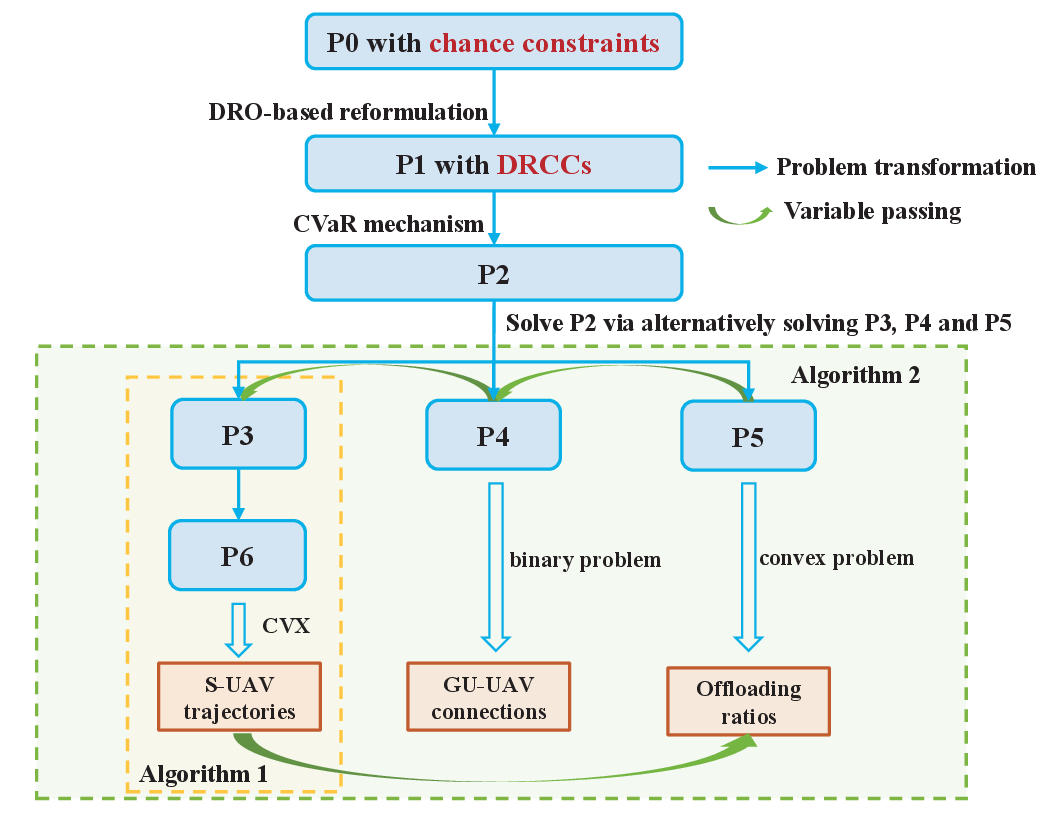}}
	\caption{Overview of the global algorithm.}
	\label{Fig-2}
\end{figure}
\subsection{Global Algorithm}\label{subsection-alg}

With the trajectory sub-problem addressed via SCA, we can integrate the above steps into a global algorithm that iteratively solves the three sub-problems to converge to a quasi-optimal solution. The global algorithm is provided in Algorithm \ref{alg-global}. A feasible initial solution is set, and with the fixed value of $\bm{w_s}$ and $\bm{\lambda}$, $\textbf{P5}$ concerning offloading decisions is a standard convex problem which can be worked out by a standard convex toolkit such as CVX (line \ref{line2-3}). With the determined $\bm{\varrho_1}^r, \bm{\varrho_2}^r,\bm{\rho}^r$ and $\bm{w_s}^r$, $\textbf{P4}$ is an integer programming problem and can be solved via optimization tools such as MOSEK (line \ref{line2-4}). Then, Algorithm \ref{alg-sca} is performed to obtain the trajectory of S-UAVs for $\textbf{P6}$ (line \ref{line2-5}). The value $\varGamma^r$ of objective function during $r$th iteration is updated and the above process is repeated until the convergence accuracy $\zeta$ is achieved (lines \ref{line2-6}-\ref{line2-8}).

\begin{algorithm}[!t]
	\caption{Joint Optimization for GU-S-UAV Association, Offloading Ratio and S-UAV Trajectory Scheduling}\label{alg-global}
	\begin{algorithmic}[1]
		\REQUIRE Parameters of the network.
		\STATE \textit{Initialization:} Set the iteration index $r=0$. Set the convergence accuracy $\zeta$\label{line2-1}. Initialize the S-UAV trajectories, the GU-UAV connection indicators, and offloading strategies.
		\REPEAT
		\STATE {With fixed $\bm{w_s}$ and $\bm{\lambda}$, solve problem $\textbf{P5}$ to obtain the optimal solution $\bm{\varrho_1}^r, \bm{\varrho_2}^r,\bm{\rho}^r$.}\label{line2-3}
		\STATE {Substitute $\bm{\varrho_1}^r, \bm{\varrho_2}^r,\bm{\rho}^r$ and $\bm{w_s}^r$ into $\textbf{P4}$ for GU-UAV connection relationships $\bm{\lambda}^r$.}\label{line2-4}
		\STATE {Update the values of optimization variables and operate Algorithm \ref{alg-sca} for $\textbf{P6}$ to obtain the S-UAV trajectories $\bm{w_s}^r$.}\label{line2-5}
		\STATE {Update the values of optimization variables and get the objective function $\varGamma^r=\sum\limits_{t=1}^T E^{total}(t)$.}\label{line2-6}
		\STATE {Update the index indicator $r\leftarrow r+1$.}\label{line2-7}
		\UNTIL {$|\varGamma^r-\varGamma^{r+1}|\leq \zeta $.}
		\ENSURE The trajectories of S-UAVs $\bm{w_s}$, GU-UAV connection relationships $\bm{\lambda}$ and offloading ratios $\bm{\rho}$.\label{line2-8}
	\end{algorithmic}
\end{algorithm}
The overview of the designed algorithms and relationships among the problems are shown in Fig. \ref{Fig-2}. The original problem $\textbf{P0}$ involving chance constraints without distribution information for the uncertain parameters is firstly reformulated into $\textbf{P1}$ via the DRO mechanism. Afterwards, by the theory of CVaR, $\textbf{P1}$ with DRCCs are further transformed into $\textbf{P2}$ through mathematical analyses. \textbf{P2} is in the form of MISOCP and comprised of three convex or approximate convex sub-problems. $\textbf{P3}$ is further transformed into $\textbf{P6}$ by means of SCA and solved by using the designed Algorithm \ref{alg-sca}. \textbf{P3}, \textbf{P4} and \textbf{P5} are solved in an alternative manner with the non-increasing objective function value after each iteration. Consequently, the objective value is obtained converging to the local optimal solution. Note that $\textbf{P5}$ is an SOCP problem with $IT$ decision variables and $10IT$ auxiliary variables, and the time complexity of one iteration for CVX to solve $\textbf{P5}$ is $\mathcal{O}\left((IT)^{3.5}\log_2{{\epsilon}^{-1}}\right)$ \cite{Computation-Zhang}. Besides, the binary variables play a dominate role in the time complexity of $\textbf{P4}$ within one iteration, given by $\mathcal{O}(2^{IMT})$. With the time complexity of $\mathcal{O}\left((MT)^{3.5}l_{it}\log_2{{\epsilon}^{-1}}\right)$ for Algorithm \ref{alg-sca}, the corresponding time complexity for global Algorithm \ref{alg-global} is $\mathcal{O}\left(r_{it}\left((IT)^{3.5}\log_2{{\epsilon}^{-1}} + 2^{IMT} + l_{it}(MT)^{3.5}\log_2{{\epsilon}^{-1}}\right)\right)$, where $r_{it}$ represents the number of iterations for convergence. Although the time complexity in solving $\textbf{P4}$ shows an exponentially increasing trend with the increment of network scale, the proposed method can converge to quasi-optimal solutions without extensive offline preparation.

\section{Simulation Results}\label{sec5}

In this section, we conduct numerical simulations to evaluate the performance of our designed algorithms. We consider a $1,000\times1,000$ m area with 3 S-UAVs and 20 time slots. The flight height of S-UAVs and the E-UAV is set to $100$ m. The coordinate of GJ is $(500 \text{ m},500\text{ m})$. The parameter settings for simulations are given in Table \ref{parameter} \cite{Energy-Zeng,Dynamic-Xu}. The implementation is accomplished by employing CVX and MOSEK as the optimization tools.
\begin{table}[!t]
	\centering
	\renewcommand\arraystretch{1.25}
	\caption{SIMULATION PARAMETERS}\label{parameter}
	\begin{tabular}{|c|c||c|c|}
	\hline
	Parameters & Value & Parameters & Value\\
	\hline
	\hline
	$M_{max}$ & 4 & $\tau$ & 2 s\\
	\hline
	$L_i(t)$ & $\left[ 1,10 \right]$ Mbits & $\bar{c}_i(t)$ & $\left[10,100\right]$ cycles/bit\\
	\hline
    $\mu_i(t)$ & 0 cycles/bit & $\sigma_i(t)$ & $0.01\bar{c}_i(t)$ \\
	\hline
	$\kappa$ & $5\times 10^{-4}$ & $\rho_0$ & $1.225$ kg/$\text{m}^3$ \\
	\hline
	$v_{bla}$ & 120 m/s & $v_{rot}$ & 4.03 m/s \\
	\hline
	$s_0$ & 0.05 & $g$ & 0.6 \\
	\hline
	$a_0$ & 0.503 $\text{m}^2$ & $v_0$ & 20 m/s\\
	\hline
	$P_1$ & 79.85 W &  $P_2$ & 88.63 W \\
	\hline
	$p_0$ & 2 W  & $p^{j}$ &20 W \\
	\hline
	$n_0$ & -174 dBm/Hz & $B_0$ & 10 MHz \\
	\hline
	$g_0$ &-50 dB & $\alpha$ & $95\%$ \\
	\hline
	$f_g$ & $10^8$ Hz & $f_u$ & $10^9$ Hz\\
	\hline
	$\varepsilon_g$ & $10^{-28}$ & $\varepsilon_u$ & $10^{-28}$ \\
	\hline
	\end{tabular}
\end{table}

\subsection{Performance of Trajectory Optimization}
Fig. \ref{f-1}(a) displays the simulation scenario, which encompasses 10 GUs, an E-UAV and a GJ within the designated area. The optimized trajectories of S-UAVs and initial trajectory of S-UAVs are shown, in which S-UAVs traverse in straight lines from their initial positions to the final destinations and they advance the same distance during each time slot. Notably, the S-UAVs tend to yaw towards their associated GUs after optimization. Besides, in some time slots, the S-UAVs prefer to maintain the hovering state, thereby reducing the energy consumption. This phenomenon shows that the proposed algorithm incorporates a comprehensive approach to balance the flight energy consumption and hovering energy consumption.

Fig. \ref{f-1}(b) evaluates the performance of the proposed SCA-based trajectory optimization approach. Compared with the baseline scheme without trajectory optimization for S-UAVs, in which S-UAVs maintain their directions towards the final points, by elaborately designing the trajectories of S-UAVs, the proposed algorithm can efficiently reduce the energy consumption of GUs and thereby, improve their QoS. This is because the transmission distances decrease when S-UAVs yew to the GUs, resulting in less transmission energy cost.
\begin{figure}
	\centering
	\subfloat[Scenario illustration and optimized trajectories of S-UAVs.]{
	\includegraphics[width=0.9\linewidth]{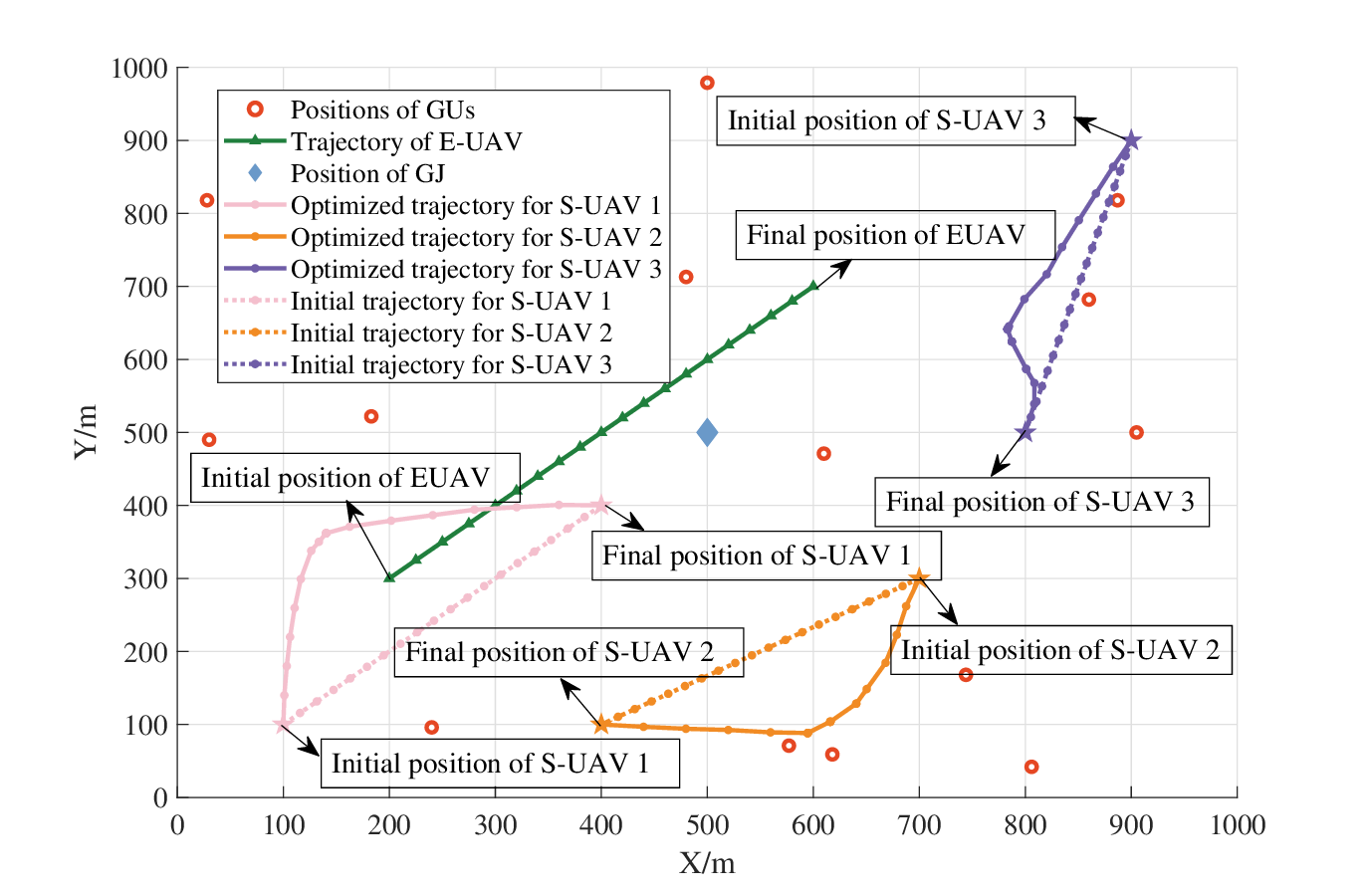}
	}
	\quad
	\subfloat[Energy cost of GUs with different S-UAV trajectories.]{
	\includegraphics[width=0.9\linewidth]{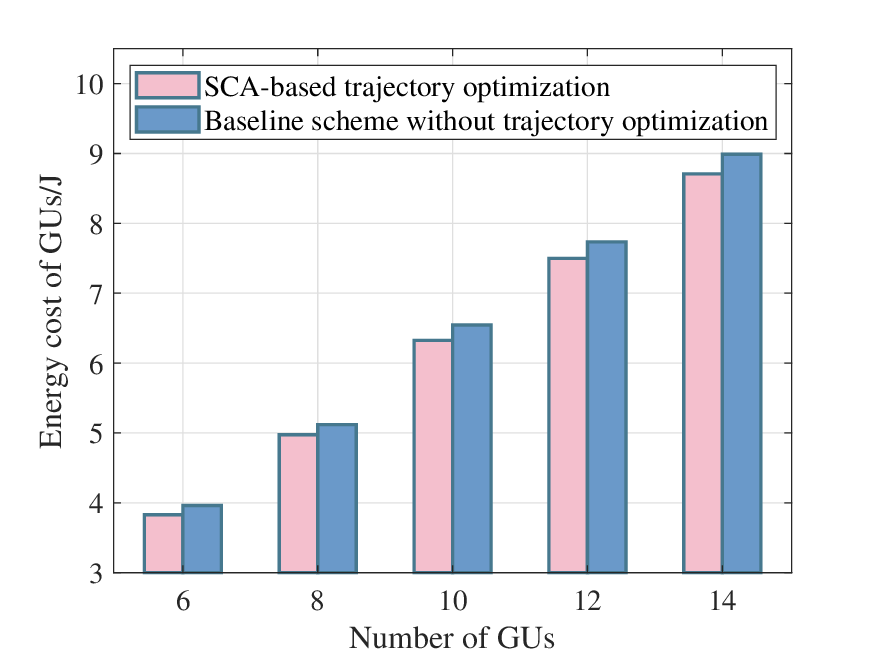}
	}
	\caption{Trajectory optimization.}\label{f-1}
\end{figure}
\subsection{Verification for Robustness}

\begin{figure*}[!t]
	\subfloat[Total Energy consumption compared with ideal circumstance.]{
	\includegraphics[width=0.30\linewidth]{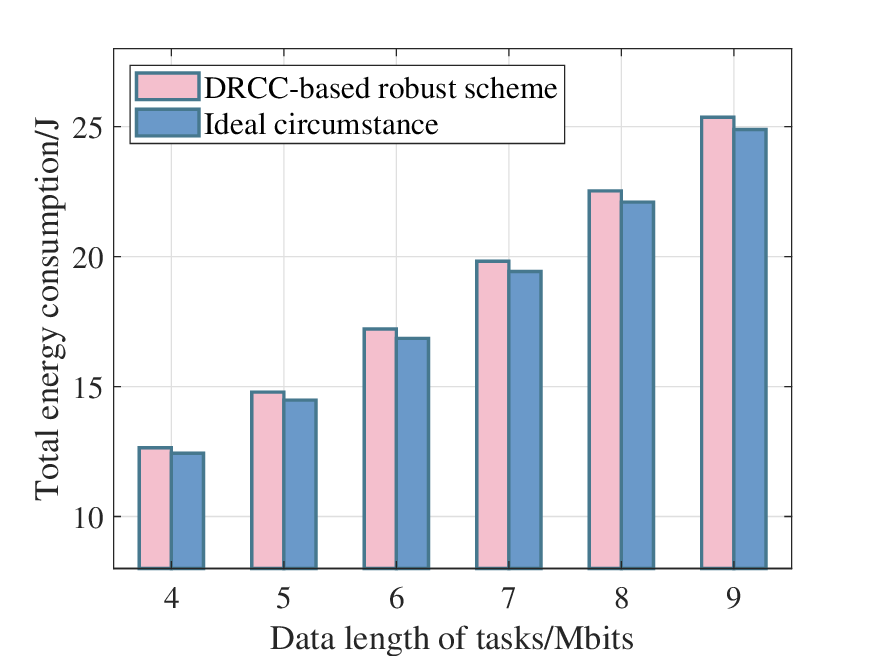}
	}
	\quad
	\subfloat[Offloaded data amount on S-UAVs compared with ideal circumstance.]{
	\includegraphics[width=0.30\linewidth]{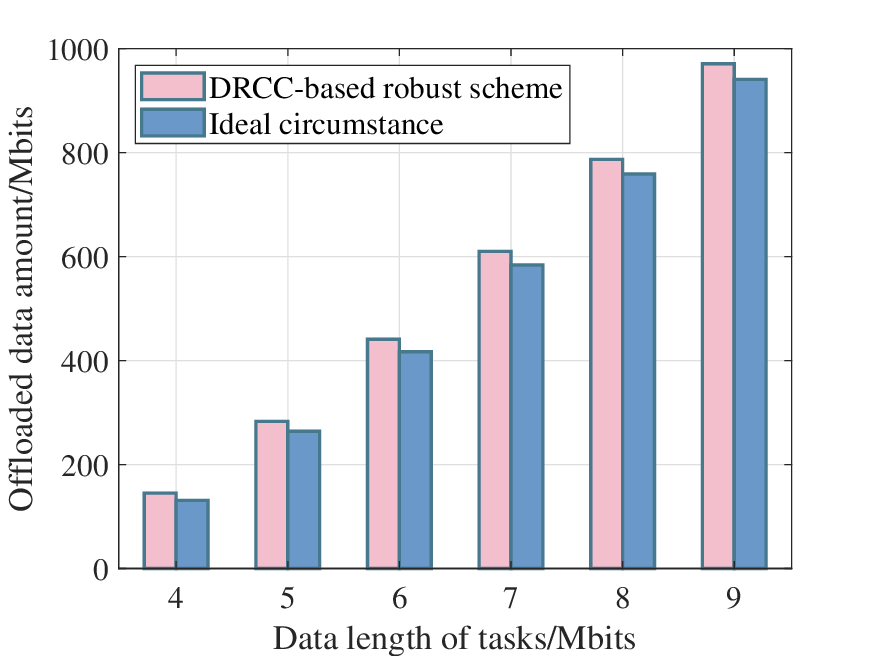}
	}
	\quad
	\subfloat[Impact of $\sigma_i(t)$ on the energy consumpion.]{
	\includegraphics[width=0.30\linewidth]{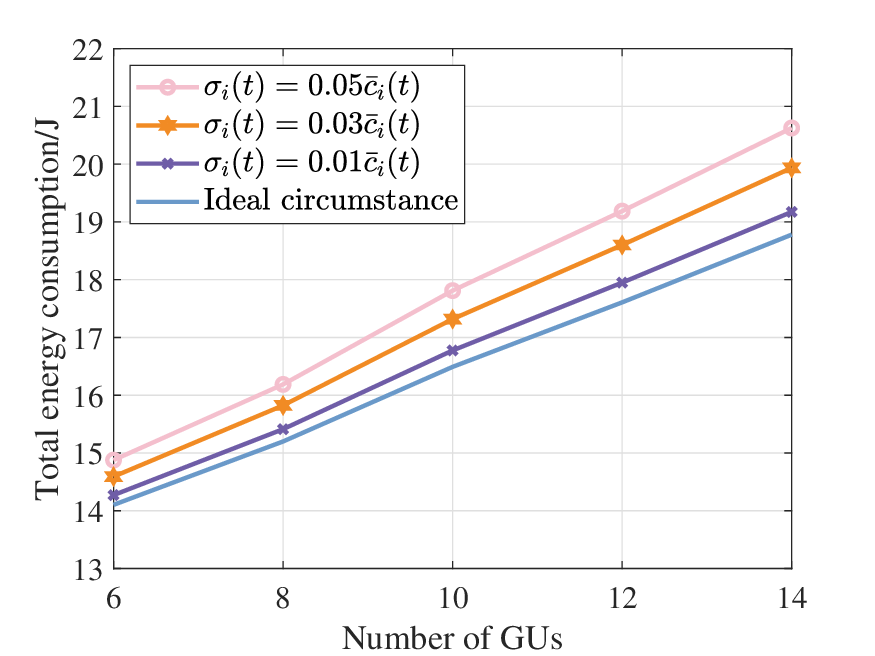}
	}
	\\
	\subfloat[Impact of $\sigma_i(t)$ on the offloaded data amount on S-UAVs.]{
	\includegraphics[width=0.30\linewidth]{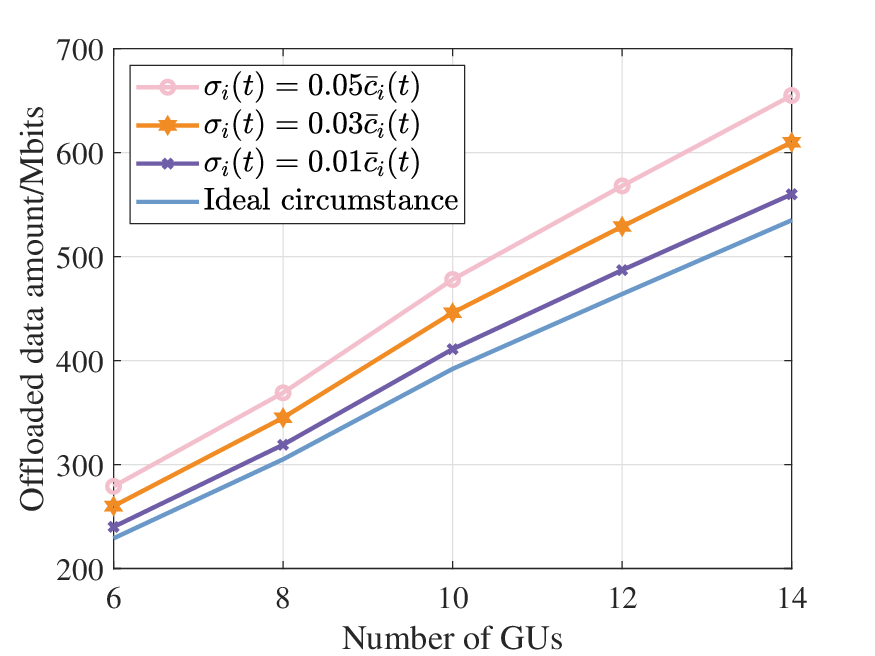}
	}
	\quad
	\subfloat[Impact of $\alpha$ on the energy consumption.]{
	\includegraphics[width=0.30\linewidth]{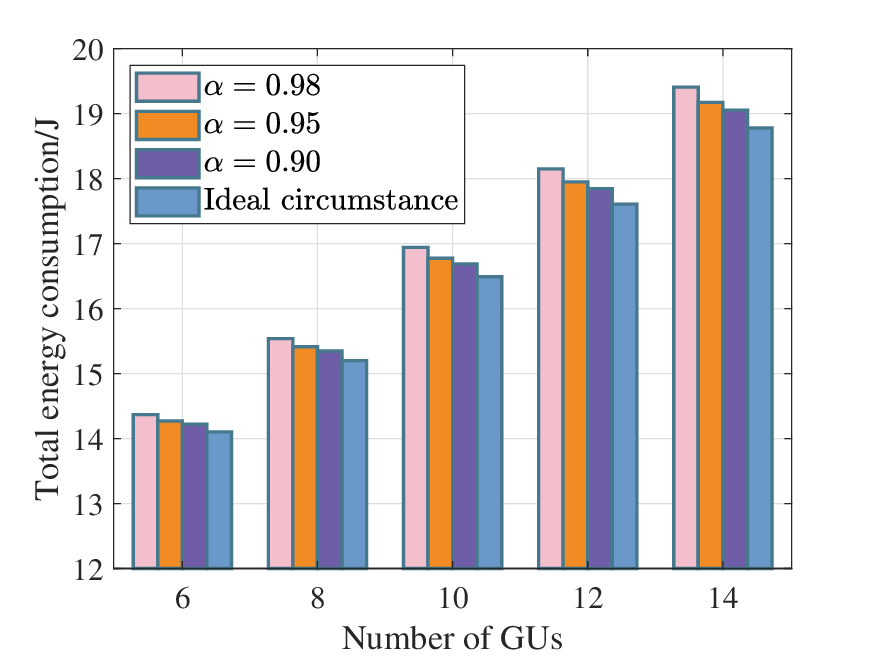}
	}
	\quad
	\subfloat[Impact of $\alpha$ on the offloaded data amoun on S-UAVs.]{
	\includegraphics[width=0.30\linewidth]{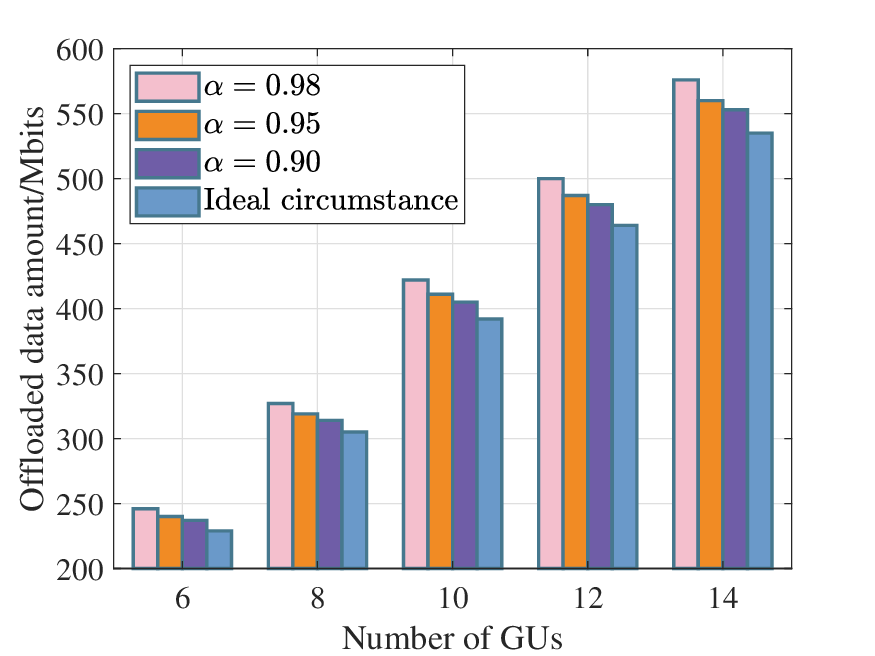}
	}
	\caption{Verification for robustness and impact of uncertain computation complexities and safety factor.}\label{f-2}
\end{figure*}
To verify the robustness of the proposed algorithm, we compare the energy cost and offloaded data amount with the results obtained under ideal conditions, as shown in Figs. \ref{f-2}(a) and (b). Specifically, in the ideal scenario, the computation complexities of tasks are obtained precisely without the influence of the uncertain task computation complexity estimation error $\Delta_i(t)$. Under such circumstances, the chance constraints involving uncertain parameters are transformed into deterministic constraints without randomness. As illustrated in Fig. \ref{f-2}(a), the total energy consumption increases with the increment of the data amount of tasks. Additionally, it is observed that the total energy consumption of the system with uncertain parameters is higher than that in the ideal cases. This is because GUs need to offload more tasks to the S-UAVs, as shown in Fig. \ref{f-2}(b), to counteract the impact of uncertain computation complexities. In contrast, the available exact computation complexity in the ideal circumstance enables GUs to offload less data and consume less energy in an optimal method without considering the worst case. With more data offloaded to the S-UAVs and more energy consumed, the task latency constraints can be guaranteed within the specified time deadline even though the system suffers from the unexpected randomness in the practical applications.
\begin{figure*}
	\centering
	\subfloat[Total energy consumption.]{
	\includegraphics[width=0.35\linewidth]{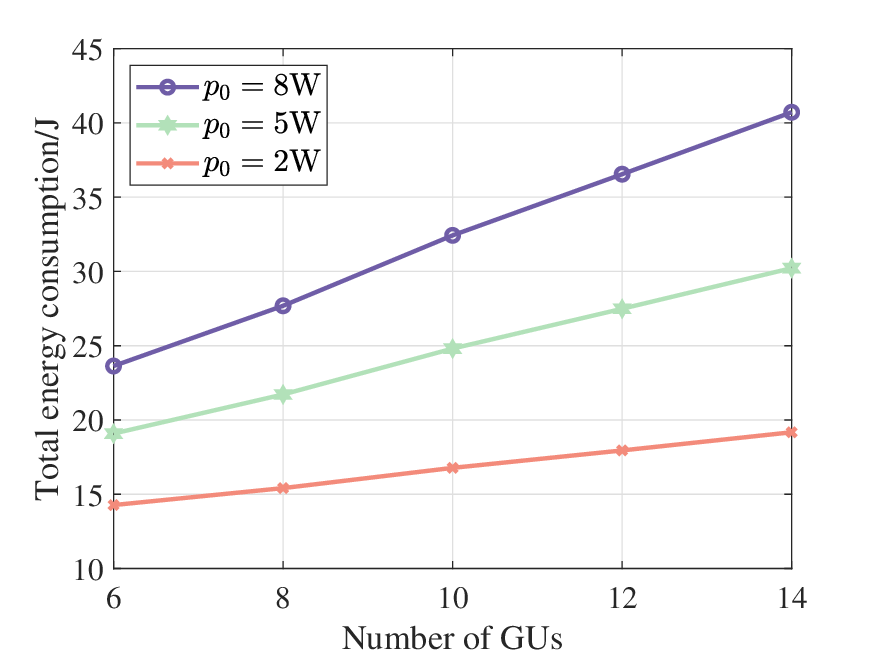}
	}
	\subfloat[Transmission energy consumption.]{
	\includegraphics[width=0.35\linewidth]{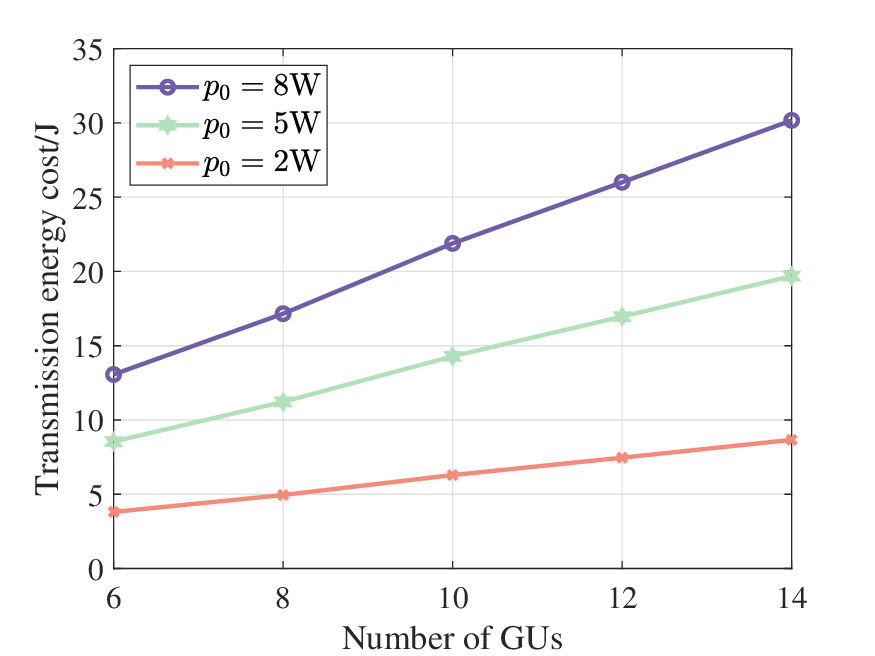}
	}
	\\
	\subfloat[Total energy consumption.]{
	\includegraphics[width=0.35\linewidth]{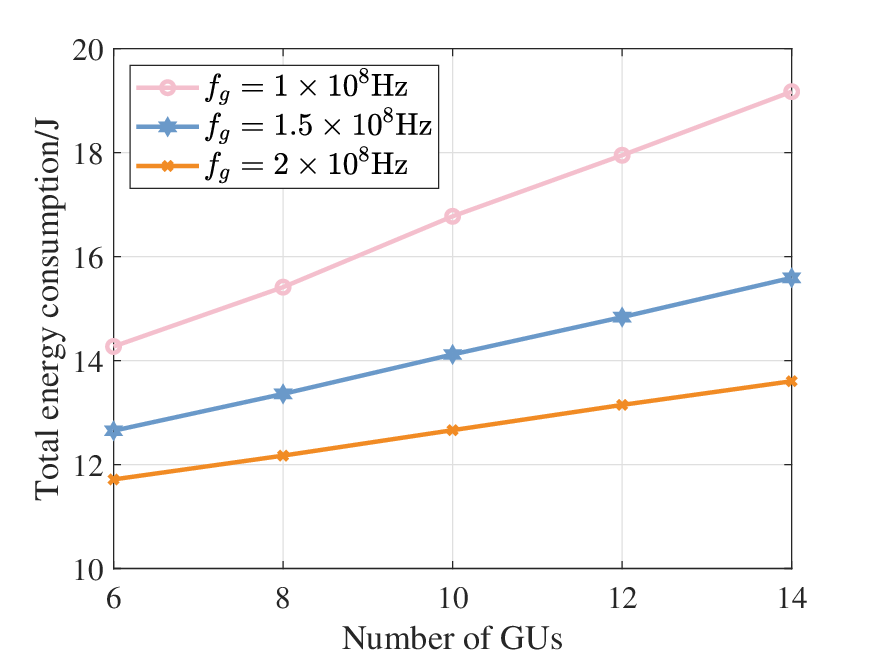}
	}
	\subfloat[Data amount.]{
	\includegraphics[width=0.35\linewidth]{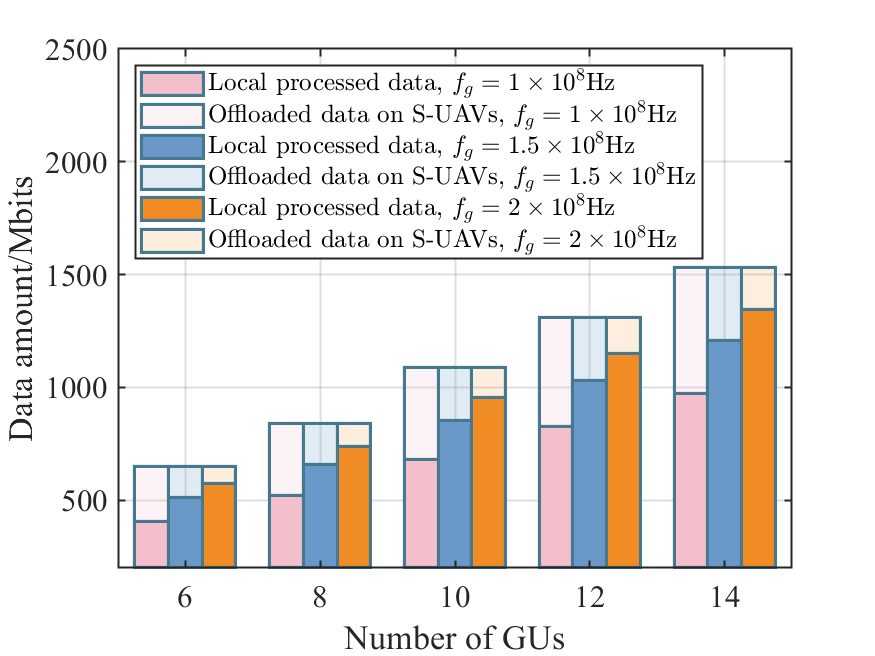}
	}
	\caption{Impact of transmission power and CPU frequency of GUs.}\label{f-6}
\end{figure*}

The effect of the potential task computation complexity estimation error on the system performance is analyzed in Figs. \ref{f-2}(c) and (d). With the increment of the standard deviation of random computation complexity $\sigma_i(t)$, more energy is consumed, as depicted in Fig. \ref{f-2}(c). This is because the uncertain computational complexities introduce more risks to the system and the larger $\sigma_i(t)$ means the uncertainties may vary more widely from the estimated value. To ensure that all tasks can be completed within the required time despite these uncertainties, GUs have to offload more data to the S-UAVs for computation assistance in the conservative manner, as shown in Fig. \ref{f-2}(d). In this way, more transmission energy is consumed, resulting in the increment of total energy cost. This phenomenon emphasizes the importance of ensuring the robustness for aerial MEC system against the random computation complexity estimation in practical applications. Besides, it is observed that the total energy consumption of the system with uncertain computation complexities is consistently higher compared with the ideal circumstances. 

Figs. \ref{f-2}(e) and (f) discuss the impact of the safety factor on the energy consumption. As illustrated in Fig. \ref{f-2}(e), the increment of the safety factor $\alpha$ brings about the rise in total energy consumption. It is due to the fact that a higher safety factor forces GUs to adopt a more conservative strategy for task offloading with more amount of data transmitted to and processed at S-UAVs, as provided in Fig. \ref{f-2}(f). In this way, while more energy is cost, the higher safety factor enhances the reliability and robustness of system against potential uncertainties in the practical applications.
\subsection{Impact of Network Parameters}
Figs. \ref{f-6}(a) and (b) analyze how the transmission power of GUs affects network performance. As depicted in Fig. \ref{f-6}(a), when the GUs offload their data to S-UAVs with higher transmission power, it takes more energy for the whole system as expected. This is because that the increment of GUs transmission results in more transmission energy consumption, as shown in Fig. \ref{f-6}(b). The effect of CPU frequency of GUs are shown in Figs. \ref{f-6}(c) and (d). As shown in Fig. \ref{f-6}(c), the stronger computing capabilities of GUs lead to less energy consumption. This is on account of the fact that GUs can process more data locally with higher efficiency and thus, with the increment of CPU frequency of GUs, more data are processed locally and less data are transmitted to the S-UAVs, as shown in Fig. \ref{f-6}(d). Hence, the transmission energy cost and the total energy cost can be decreased efficiently.

\section{Conclusion}\label{sec6}

In this paper, we focused on the multi-UAV enabled secure offloading system, including S-UAVs serving as aerial BSs and an E-UAV overhearing the communication links. Owing to the limited battery capacities, we further studied the minimization of the energy consumption of both GUs and S-UAVs. Since the unpredictable task computation complexities without distribution information brought up great challenges to the performance of the network, we optimized the trajectories of S-UAVs, connection relationships between GUs and UAVs, as well as offloading ratios under the worst case. Firstly, the DRO method was employed for the conservative solutions considering all possible distributions of the uncertainties. Then, the concept on CVaR was introduced, which successfully transformed the DRCCs into MISOCP constraints. The problem was further decomposed into three sub-problems and the SCA-based trajectory optimization method was proposed. By solving the three sub-problems alternatively, the quasi-optimal solutions were obtained. Finally, the numerical simulations were conduced to verify the robustness against the uncertainties with 2\% more energy consumed compared with the ideal circumstances. In future works, the collaborative operations of various low-altitude aircraft as supplementary edge nodes will be explored to address the coverage limitations of UAVs. Besides, the sophisticated hybrid algorithms are expected to cope with the multiple potential uncertainties in practical environments.
\begin{appendices}
	
\section{}\label{appendix}

For the measurable loss function $\phi(\xi)$, CVaR under safety factor $\alpha$ and probability distribution $\mathbb{P}$ is defined as
\begin{equation}
	\mathbb{P}-CVaR_\alpha(\phi(\xi))=\underset{\beta\in\mathbb{R}}{inf}\left\{\beta+\frac{1}{1-\alpha}\mathbb{E_P}\left\{\phi(\xi)-\beta\right\},0\right\},
\end{equation}
where $\beta$ is an auxiliary variable on $\mathbb{R}$. Hence, the worst-case CVaR for the loss function $\phi(\xi)=\Theta\xi+\theta^0$ is expressed as
\begin{equation}
	\underset{\mathbb{P}\in\mathcal{P}}{sup}\quad\underset{\beta\in\mathbb{R}}{inf}\left\{\beta+\frac{1}{1-\alpha}\mathbb{E_P}\left[\Theta\xi+\theta^0-\beta\right]^+\right\},
\end{equation}
in which $\underset{\mathbb{P}\in \mathcal{P}}{sup}$ represents the upper bound of possible distribution under $\mathcal{P}$, and can be further transformed into $\underset{\beta\in\mathbb{R}}{inf}\left\{\beta+\frac{1}{1-\alpha}\underset{\mathbb{P}\in\mathcal{P}}{sup}\quad\mathbb{E_P}\left[\Theta\xi+\theta^0-\beta\right]^+\right\}$ according to the Saddle Point Theorem. Based on the uncertainty set $\mathcal{P}=\left\{\mathbb{P}\in\mathcal{P}|\mathbb{E_P}(\xi)=\mu_\xi,\mathbb{D_P}(\xi)=\sigma_\xi^{2}\right\}$, the inner layer optimization $\underset{\mathbb{P}\in\mathcal{P}}{sup}\quad\mathbb{E_P}\left[\Theta\xi+\theta^0-\beta\right]^+$ is equivalent to
\begin{equation}\label{e59}
	\begin{split}
		&\underset{\varphi\in\mathcal{C}}{sup}\quad\int_\mathbb{R}\left(\Theta\xi+\theta^0-\beta\right)^+\varphi\,d\varrho\\
		\textrm{s.t.}\quad&\int_\mathbb{R}\varphi\,d\varrho=1,\int_\mathbb{R}\varrho\varphi\,d\varrho=\Theta\mu_\xi,\\
		&\int_\mathbb{R}\varrho^2\varphi\,d\varrho=\Theta^2\mu_\xi^2+\Theta^2\sigma_\xi^2,
	\end{split}
\end{equation}
where $\varrho =\Theta\xi$. $\varphi$ denotes the decision variables and $\mathcal{C}$ represents the cone of nonnegative Borel measures on $\mathbb{R}$. By introducing the dual variables $\chi_1$, $\chi_2$ and $\chi_3$ of the constraints in (\ref{e59}), the Lagrangian dual problem is obtained as
\begin{align}
	\underset{\chi_1,\chi_2,\chi_3}{inf}&\quad \chi_1+\chi_2\Theta\mu_\xi+\chi_3\left(\Theta^2\mu_\xi^2+\Theta^2\sigma_\xi^2\right)\nonumber\\
	\textrm{s.t.}\quad&\chi_1+\chi_2 \varrho+\chi_3\varrho^2\geq0,\label{e60}\\
	&\chi_1+\chi_2\varrho+\chi_3\varrho^2\geq\Theta+\varrho-\beta,\label{e61}\\
	&\chi_3>0.
\end{align}
Then, $\varrho^*=-\frac{\chi_2}{2\chi_3}$ is substituted into (\ref{e60}) and (\ref{e61}). According to the Strong Duality Theorem, (\ref{e60}) and (\ref{e61}) can be equivalently transformed into
\begin{equation}
	\chi_1-\frac{\chi_2^2}{4\chi_3}\geq 0,
\end{equation}
and
\begin{equation}
	\chi_1+\beta-\Theta-\frac{\left(\chi_2-1\right)^2}{4\chi_3}\geq 0.
\end{equation}
With $\chi_1=e+\frac{1}{4z}\left(q-\Theta\mu_\xi\right)^2$, $\chi_2=\frac{q-\Theta\mu_\xi}{2z}$ and $\chi_3=\frac{1}{4z}$, Eq. (\ref{e59}) is reformulated as
\begin{equation}
	\begin{split}
		\underset{\beta,e,q,z,s}{inf}& e+s ,\\
		&e-\theta ^0+\beta +q-\Theta \mu_\xi -z>0,\\
		&e\geq 0,z> 0,\\
		&4zs\geq q^2+\Theta^2\sigma_\xi^2.
	\end{split}
\end{equation}
Consequently, regarding the definition of CVaR, the worst-case CVaR for the loss function $\phi(\xi)=\Theta\xi+\theta^0$ under safety factor $\alpha$ can be approximately transformed into an SOCP problem, i.e.,
\begin{equation}
	\begin{split}
		\underset{\beta,e,q,z,s}{inf}&\beta +\frac{1}{1-\alpha}\left( e+s \right) ,\\
		&e-\theta ^0+\beta +q-\Theta \mu_\xi -z>0,\\
		&e\geq 0,z> 0,\\
		&\begin{Vmatrix}
				q\\
				\Theta\sigma_\xi\\
				z-s
		\end{Vmatrix} \leq z+s.\\
	\end{split}
\end{equation}
\end{appendices}
\bibliographystyle{IEEEtran}
\bibliography{reference.bib}

\begin{IEEEbiography}[{\includegraphics[width=1in,height=1.25in,clip,keepaspectratio]{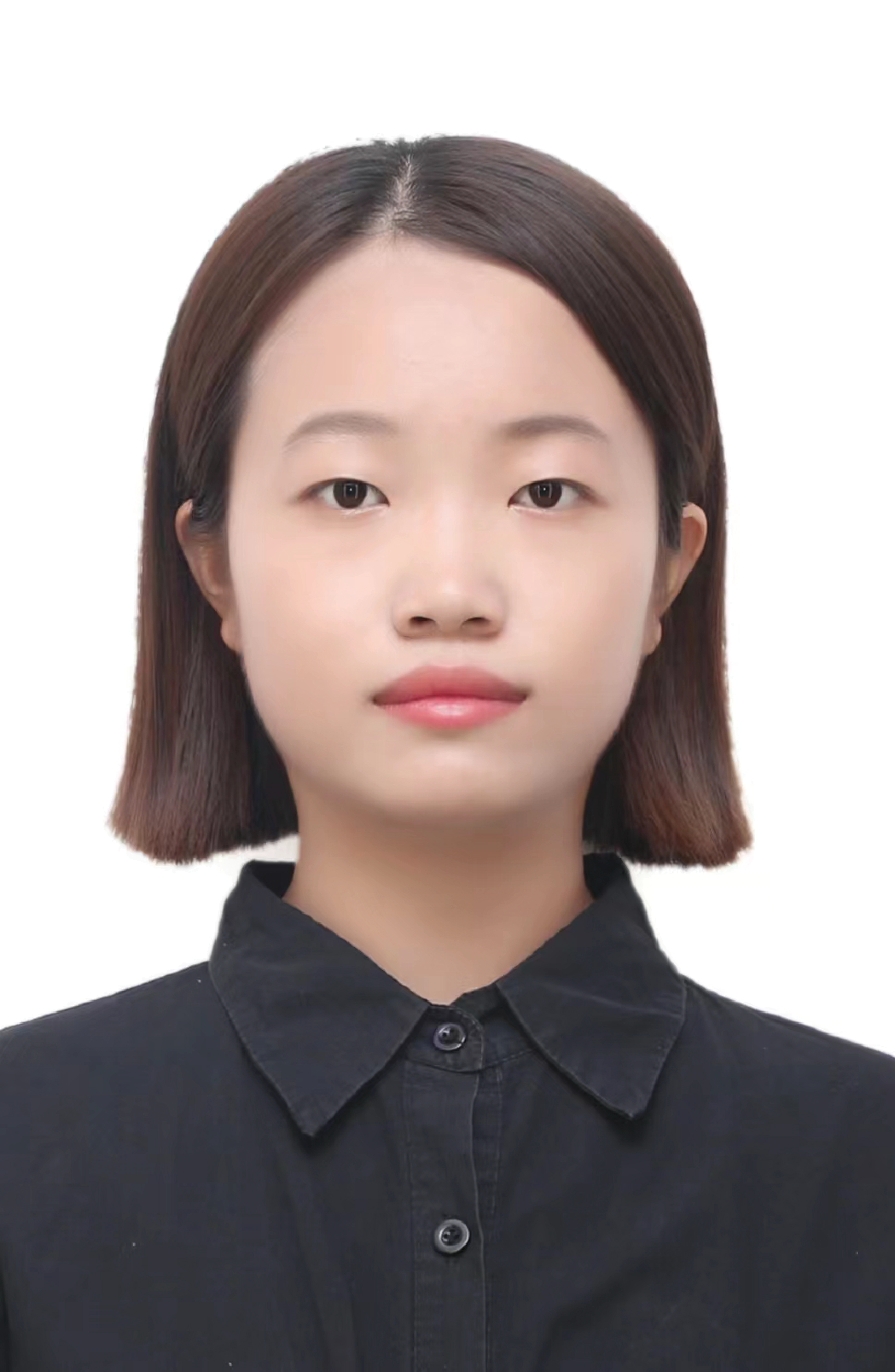}}] {Can Cui} is a postgraduate student with the College of Electronic and Information Engineering, Nanjing University of Aeronautics and Astronautics, Nanjing, China. Her current research interests include convex optimization and its applications in computation offloading and resource allocation, edge computing, and low-altitude intelligent networks.
\end{IEEEbiography}

\begin{IEEEbiography}[{\includegraphics[width=1in,height=1.25in,clip,keepaspectratio]{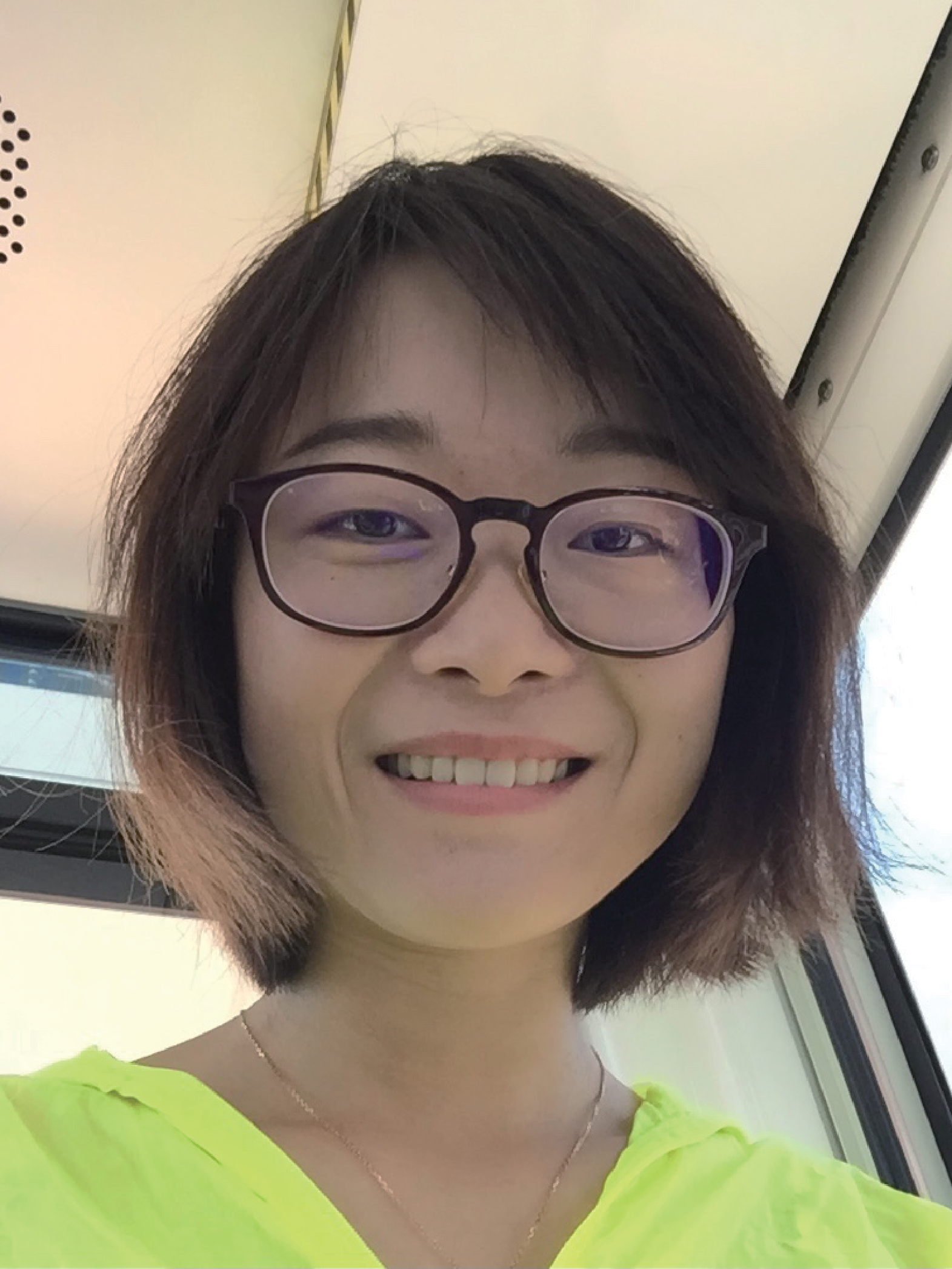}}] {Ziye Jia} (Member, IEEE) received the B.E., M.S., and Ph.D. degrees in communication and information systems from Xidian University, Xi'an, China, in 2012, 2015, and 2021, respectively. From 2018 to 2020, she was a Visiting Ph.D. Student with the Department of Electrical and Computer Engineering, University of Houston. She is currently an Associate Professor with the Key Laboratory of Dynamic Cognitive System of Electromagnetic Spectrum Space, Ministry of Industry and Information Technology, Nanjing University of Aeronautics and Astronautics, Nanjing, China. Her current research interests include space-air-ground networks, aerial access networks, UAV networking, resource optimization,  machine learning, etc.
\end{IEEEbiography}

\begin{IEEEbiography}[{\includegraphics[width=1in,height=1.25in,clip,keepaspectratio]{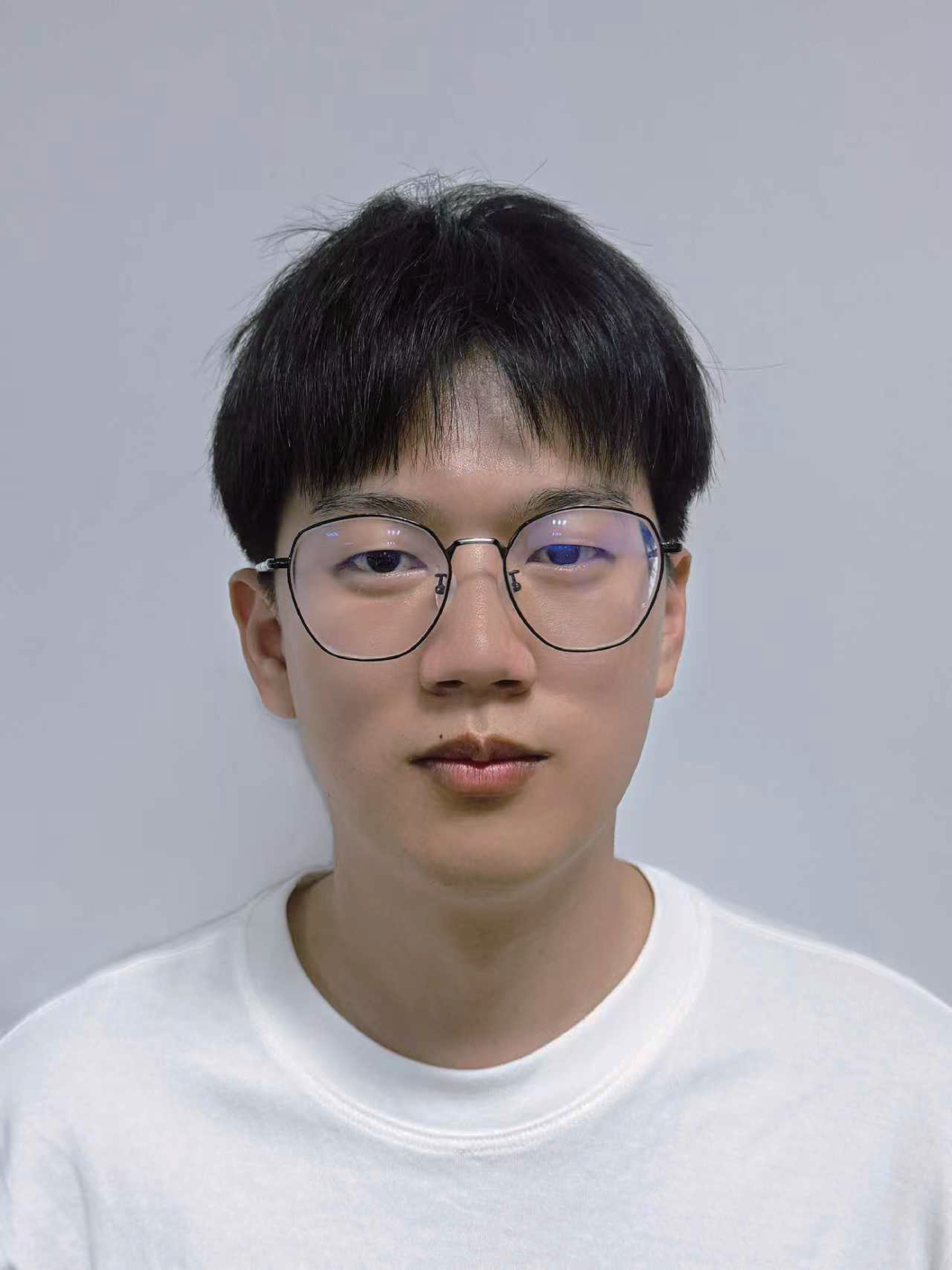}}] {Jiahao You} is currently pursuing the Ph.D. degree with the School of Electronic and Information Engineering, Nanjing University of Aeronautics and Astronautics. His current research interests include deep reinforcement learning and its applications in computation offloading and resource allocation, edge computing, low-altitude intelligent networks, and UAV trajectory planning.
\end{IEEEbiography}

\begin{IEEEbiography}[{\includegraphics[width=1in,height=1.25in,clip,keepaspectratio]{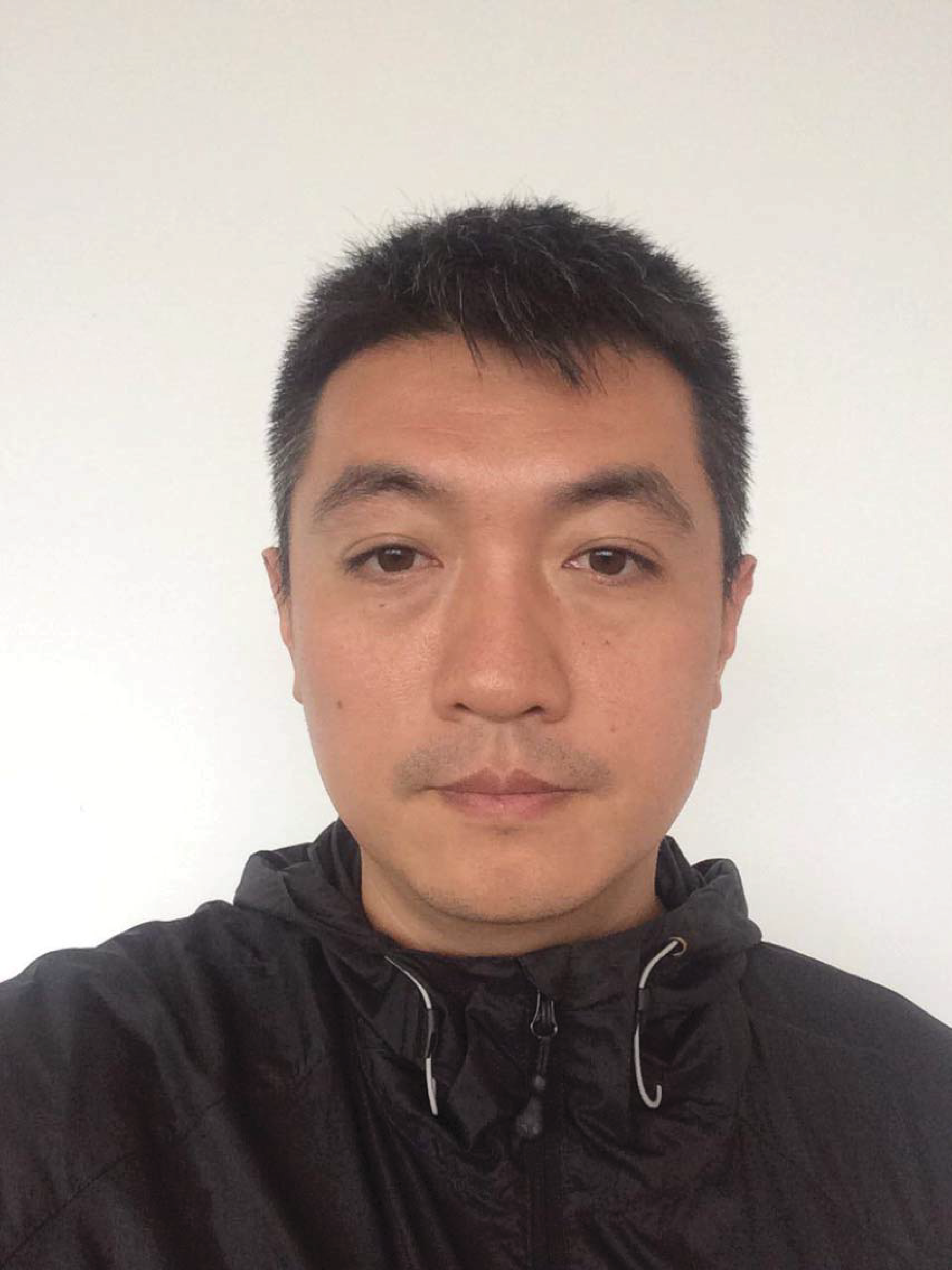}}]{Chao Dong} (Senior Member, IEEE) received his Ph.D degree in Communication Engineering from PLA University of Science and Technology, China, in 2007. He is now a full professor with College of Electronic and Information Engineering, Nanjing University of Aeronautics and Astronautics, China. His current research interests include D2D communications, UAVs swarm networking and anti-jamming network protocol.
\end{IEEEbiography}

\begin{IEEEbiography}[{\includegraphics[width=1in,height=1.25in,clip,keepaspectratio]{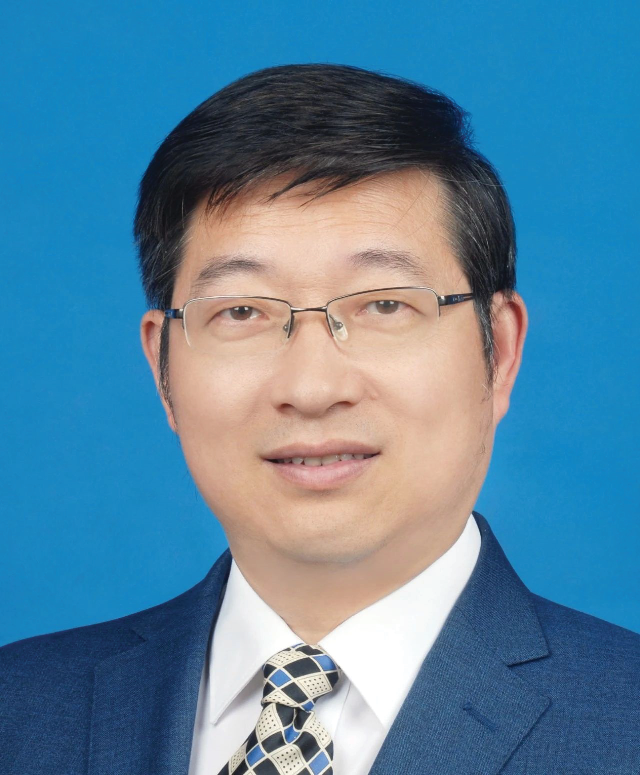}}] {Qihui Wu} (Fellow, IEEE) received the B.S. degree in communications engineering and the M.S. and Ph.D. degrees in communications and information systems from the Institute of Communications Engineering, Nanjing, China, in 1994, 1997, and 2000, respectively. From 2003 to 2005, he was a Post-Doctoral Research Associate with Southeast University, Nanjing. From 2005 to 2007, he was an Associate Professor with the College of Communications Engineering, PLA University of Science and Technology, Nanjing, where he was a Full Professor, from 2008 to 2016. From March 2011 to September 2011, he was an Advanced Visiting Scholar with the Stevens Institute of Technology, Hoboken, NJ, USA. Since May 2016, he has been a Full Professor with the College of Electronic and Information Engineering, Nanjing University of Aeronautics and Astronautics, Nanjing. His current research interests include wireless communications and statistical signal processing, with an emphasis on system design of software defined radio, cognitive radio, and smart radio.
\end{IEEEbiography}

\begin{IEEEbiography}[{\includegraphics[width=1in,height=1.25in,clip,keepaspectratio]{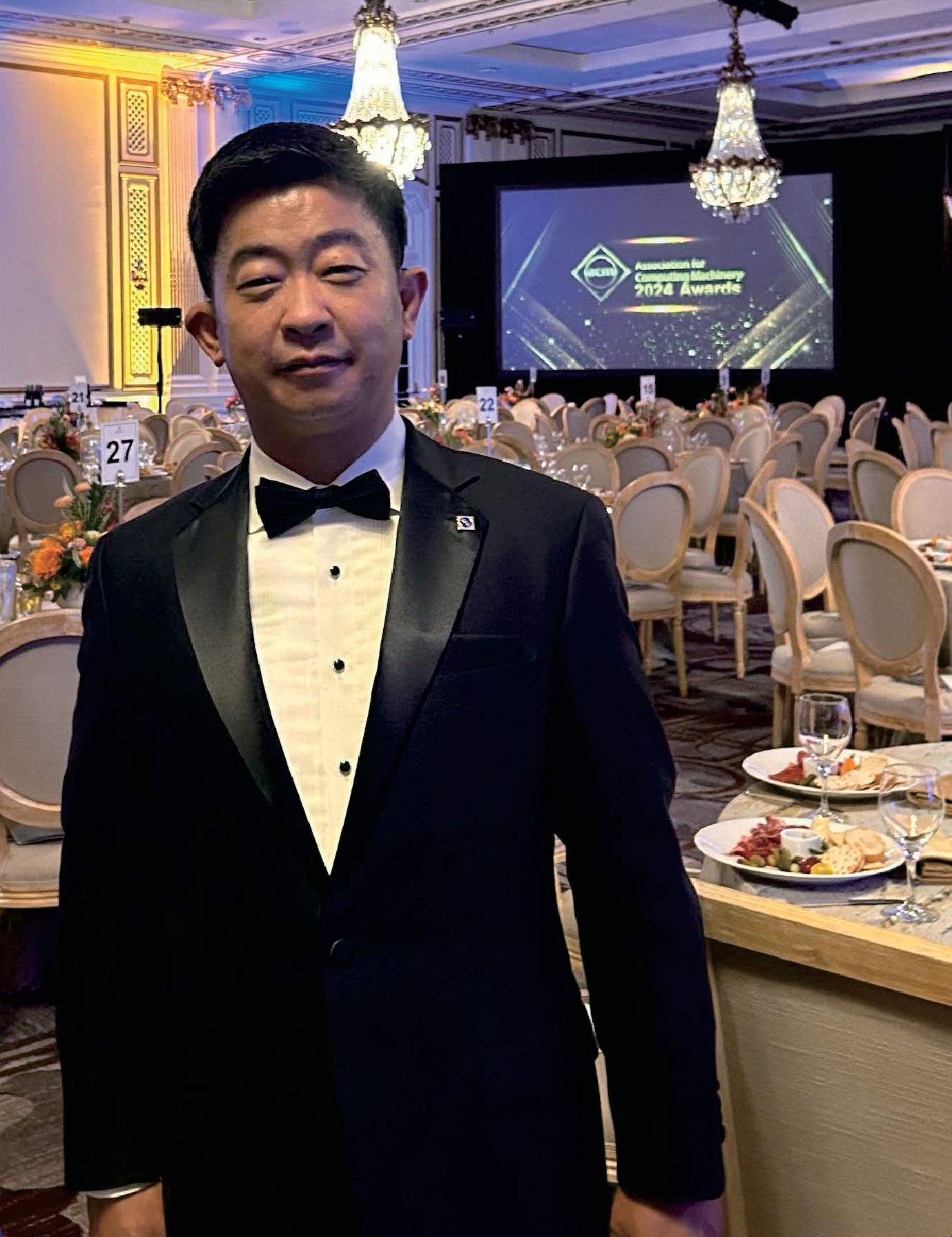}}] {Zhu Han} (Fellow, IEEE) (S'01-M'04-SM'09-F'14) received the B.S. degree in electronic engineering from Tsinghua University, in 1997, and the M.S. and Ph.D. degrees in electrical and computer engineering from the University of Maryland, College Park, in 1999 and 2003, respectively. From 2000 to 2002, he was an R\&D Engineer of JDSU, Germantown, Maryland. From 2003 to 2006, he was a Research Associate at the University of Maryland. From 2006 to 2008, he was an assistant professor at Boise State University, Idaho. Currently, he is a John and Rebecca Moores Professor in the Electrical and Computer Engineering Department as well as in the Computer Science Department at the University of Houston, Texas. Dr. Han's main research targets on the novel game-theory related concepts critical to enabling efficient and distributive use of wireless networks with limited resources. His other research interests include wireless resource allocation and management, wireless communications and networking, quantum computing, data science, smart grid, carbon neutralization, security and privacy.  Dr. Han received an NSF Career Award in 2010, the Fred W. Ellersick Prize of the IEEE Communication Society in 2011, the EURASIP Best Paper Award for the Journal on Advances in Signal Processing in 2015, IEEE Leonard G. Abraham Prize in the field of Communications Systems (best paper award in IEEE JSAC) in 2016, IEEE Vehicular Technology Society 2022 Best Land Transportation Paper Award, and several best paper awards in IEEE conferences. Dr. Han was an IEEE Communications Society Distinguished Lecturer from 2015 to 2018 and ACM Distinguished Speaker from 2022 to 2025, AAAS fellow since 2019, and ACM Fellow since 2024. Dr. Han is a 1\% highly cited researcher since 2017 according to Web of Science. Dr. Han is also the winner of the 2021 IEEE Kiyo Tomiyasu Award (an IEEE Field Award), for outstanding early to mid-career contributions to technologies holding the promise of innovative applications, with the following citation: ``for contributions to game theory and distributed management of autonomous communication networks."
\end{IEEEbiography}

\end{document}